\documentclass[12pt]{article}
\usepackage{fullpage,setspace}
\usepackage{amssymb, amsthm, amsmath}
\usepackage{amsmath, amssymb}
\usepackage{graphicx}
\usepackage[authoryear]{natbib}
\usepackage{bm}
\usepackage{algorithm,algpseudocode}
\usepackage{hyperref}
\usepackage{mathabx}
\usepackage{mathtools}
\usepackage{float}
\usepackage{url}
\usepackage{color}

\usepackage{graphicx}
\usepackage{lineno}
\usepackage{times}
\usepackage{algorithm,algpseudocode}
\usepackage{hyperref}
\usepackage{multirow}
\usepackage{multirow}
\usepackage{booktabs}
\usepackage{cleveref}
\usepackage{makecell}

\usepackage[T1,OT1]{fontenc}

\usepackage{caption}
\newtheorem{definition}{Definition}
\newtheorem{theorem}{Theorem}
\newtheorem{lemma}{Lemma}

\usepackage{fancyhdr}
\date{}
\usepackage{authblk}

\title{Hypothesis Testing for Nonlinear and Interactive Effects Using Additive Gaussian Processes}

\author[1]{Nate Wiecha\thanks{Corresponding Author; email: nwiecha@unc.edu}}
\author[2]{Cailee E. Harrington}
\author[1]{Ameer Qaqish}
\author[1]{Haibo Zhou}
\author[2]{Marc Serre}
\author[2]{Lauren A. Eaves}
\author[2]{Rebecca Fry}
\author[1]{Didong Li}

\affil[1]{Department of Biostatistics, University of North Carolina at Chapel Hill}
\affil[2]{Department of Environmental Sciences and Engineering, University of North Carolina at Chapel Hill}
\begin{document}

\label{firstpage}

\maketitle

\begin{abstract}
\setstretch{1.2}

Real-world data often exhibit nonlinear, interactive associations between predictors and an outcome. For example, in public health studies, environmental exposures and demographic variables can have nonlinear associations with health outcomes and modify the associations of other predictors. Parametric regression models facilitate interpretable inference and hypothesis testing with Type I error rate control, but may be misspecified when associations are nonlinear. In contrast, flexible nonparametric regression models such as Gaussian processes (GPs) can capture complex nonlinearity, but are often more difficult to interpret and do not readily provide hypothesis tests with Type I error rate control. To bridge this gap, we propose a hypothesis testing framework based on additive GPs restricted to pairwise interactions, allowing tests of variable associations and interaction effects with Type I error rate control. We further establish contraction rates for additive GPs, providing theoretical support for their reduced dependence on predictor dimension when higher-order interactions are absent. Through simulations, we show that restricting the model to pairwise interactions can substantially improve testing power relative to methods based on standard GPs. Applications to environmental health studies identify nonlinear associations between metal exposures and health outcomes as well as effect modification by demographic factors.

\noindent\textbf{Keywords:} Interaction, additive, Gaussian process, hypothesis testing, contraction rate, nonparametric regression.
\end{abstract}

\setstretch{1.2}
\clearpage
\section{Introduction}\label{s:intro}


Hypothesis testing is fundamental to statistical inference, providing a principled framework for assessing evidence for associations between predictors and an outcome. Parametric regression models provide a natural framework for such testing, supported by a well-developed literature on statistical inference, interpretable model parameters, and hypothesis tests with Type I error rate control~\citep{casella2024statistical, boos2013essential, agresti2015foundations}. However, relationships in real-world data can be substantially more complex than those specified by standard parametric models. For example, in public health studies, environmental exposures, nutritional intakes, and demographic factors may have nonlinear associations with a health outcome and modify the associations of other exposures~\citep{bobb2015bayesian, liu2022cross, ngiangi_bakwin2001semiparametric,Morgensterne037860}. When the assumed functional form is misspecified, parametric tests can lose power or provide misleading inference. Flexible, nonparametric regression methods, including smoothing splines and Gaussian processes (GPs), provide an alternative by allowing complex nonlinear relationships to be learned from the data \citep{wood2017generalized, bobb2015bayesian}.

While nonparametric regression models make fewer assumptions about the functional form of associations, formal hypothesis tests of variable-level effects with Type I error rate control are considerably less straightforward than in parametric regression. Many approaches instead quantify variable importance, allowing researchers to rank predictors or characterize their associations with the outcome. Examples include Sobol indices \citep{sobol1990sensitivity}, which has been used with GP models \citep[e.g.][]{lu2022oak}; broader overviews of variable importance measures are given in \cite{molnar2025}and  \cite{hastie2009elements}. However, measures of variable importance do not provide formal evidence for whether a predictor or interaction is associated with the outcomes. This distinction is particularly important in public health studies, where the scientific question is often not which exposure is more important, but whether there is sufficient evidence that a particular exposure is associated with a health outcome or that one exposure modifies the effect of another. As a result, formal hypothesis testing provides an important complement to variable-importance measures.

A general nonparametric regression model for a Gaussian outcome $y$ can be written as
\begin{align}
    y_i = f_0(\mathbf{x}_i) + \epsilon_i,~~
    \epsilon_i \sim \text{ i.i.d. } N(0, \sigma^2),\label{mod:np}
\end{align}
where $f_0$ is an unknown function and $\mathbf{x}_i\in\mathbb{R}^p$ is a vector of predictors.  $f_0$ can be modeled using flexible approaches such as GP regression \citep{williams2006gaussian}, smoothing splines \citep{wood2017generalized}, or other machine learning models such as random forests or deep neural networks \citep{breiman2001random, hastie2009elements}. Among these approaches, GP regression is particularly amenable to hypothesis testing because hypotheses about $f_0$ can be formulated through parameters of its covariance kernel. Consequently, several GP-based methods have been developed for assessing associations and interactions, particularly in public health applications. For instance, \citet{maity2011powerful} proposed garrote kernel machines (GKM) for testing association between genetic predictors and a health outcome using  GPs, which can also be referred to as kernel machine regression models. \cite{bobb2015bayesian} proposed Bayesian Kernel Machine Regression (BKMR), a Bayesian alternative for estimating exposure-response relationships, which has seen wide use. \cite{liu2022cross} proposed a cross-validated ensemble of kernels (CVEK), extending GKM to test for nonlinear interactions between groups of variables. 

Existing methods have several drawbacks. GKM focuses on a standard GP regression model, which allows arbitrarily high orders of interaction between all $p$ predictors. Consequently, this model suffers from the curse of dimensionality, as the difficulty of estimating high-order interactions increases exponentially with $p$. Additionally, GKM turns out to be a test of {\it linear} association, due to a subtle issue with the hypothesis test (see Supplemental Material Section \ref{s:supp_gkm_lin} for details). BKMR  uses Bayesian variable selection with GP regression and also allows high-order interactions, increasing model complexity when such interactions are not necessary. Moreover, its Bayesian variable selection procedure does not provide frequentist hypothesis tests with Type I error rate control. The grouped variable selection, designed to alleviate effects of collinearity, can also provide misleading results \citep{wiecha2026novel}. CVEK focuses specifically on testing nonlinear interactions between groups of variables and therefore does not provide a general framework for testing variable importance.

A central issue underlying these limitations is that the general nonparametric model (\cref{mod:np}) can be overly flexible when $p$ is moderately large. While flexibility is needed to capture nonlinear interactions, an unrestricted $f_0$ allows interactions of arbitrarily high order, which may be unnecessary and can lead to the curse of dimensionality and more difficult interpretation. This motivates more structured nonparametric models that impose reasonable restrictions on $f_0$ while retaining flexibility. In particular, \cite{duvenaud2011additivegaussianprocesses}, and shortly afterward \cite{durrande2012additive}, proposed additive GP regression, which includes as a special case the model we consider:
\begin{align}\label{mod:pairwise_np}
    y_i = \sum_{j} f_j(x_{ij}) + \sum_{j<k}f_{j,k}(x_{ij}, x_{ik}) + \epsilon_i,~~\epsilon_i\sim \text{ i.i.d. } N(0, \sigma^2),
\end{align}
for $j, k=1, ..., p$, which the regression function $f_0$ is decomposed into nonlinear main effects and pairwise interaction effects, with each component modeled using a GP. Such a restriction can improve predictive performance when higher-order interactions are absent \citep{duvenaud2011additivegaussianprocesses} and also facilitates interpretation, since continuous nonlinear interactions involving many predictors are inherently difficult to interpret. In fact, \cite{lu2022oak} found additive GPs with second-order interaction can usually provide comparable predictions to models including higher-order terms. More importantly for hypothesis testing, reducing unnecessary model complexity can lead to greater power when the true regression function contains only low-order interactions. Building on this structured GP model, we propose hypothesis tests for associations between individual predictors or groups of predictors and the outcome, as well as for pairwise interaction effects.

Our contributions are fourfold. First, we develop a general hypothesis testing framework based on the pairwise additive model (\cref{mod:pairwise_np}) for testing associations between individual predictors or groups of predictors and an outcome, as well as specific or overall pairwise interaction effects. Second, we establish contraction rates for additive GPs and show that, under appropriate conditions, the rate depends on the dimension of the individual additive components rather than the total number of predictors, thereby alleviating the curse of dimensionality. Third, through simulations, we show that the proposed tests can achieve substantially higher power than tests based on standard GPs when higher-order interactions are absent. Finally, we demonstrate the scientific utility of the proposed framework through substantive public health applications. Using private-well measurements from 1,693 North Carolina census tracts, we find that inorganic arsenic concentrations are associated with obesity prevalence, and the association may be modified by income, and in NHANES, we identify nonlinear associations between metal exposures and liver injury and evidence that the association between manganese and liver health may differ by sex.

\section{Background: Gaussian Process (GP)}

We review GP regression, which is the basis of the aforementioned existing methods, including GKM, BKMR, and CVEK, as well as our proposed method. A GP is a stochastic process whose finite realizations follow multivariate Gaussian distributions.

\begin{definition}
A stochastic process $W$ over $\mathcal{X}$ follows a GP, denoted by \(W\sim GP(m,K)\), if, for any $\mathbf{x}_1,\ldots,\mathbf{x}_n\in\mathcal{X}$,
\(
\left[W(\mathbf{x}_1),\ldots,W(\mathbf{x}_n)\right]^\top
\sim N(\boldsymbol{\mu},\boldsymbol{\Sigma}),
\)
where $\mu_i=m(\mathbf{x}_i)$ and $\Sigma_{ij}=K(\mathbf{x}_i,\mathbf{x}_j)$. Here, $m$ is called the mean function and $K$ is called the covariance or kernel function.
\end{definition}

Without loss of generality, we assume the mean function is zero throughout this paper. The covariance function is typically chosen from a parametric family $K_{\boldsymbol{\theta}}$, where $\boldsymbol{\theta}$ denotes the covariance parameters. Letting $d=\|\mathbf{x}-\mathbf{x}'\|_2$, two commonly used covariance functions are the squared exponential covariance
\(
K_{\boldsymbol{\theta}}(d)
=
\lambda\exp\left(-\frac{1}{\rho}d^2\right),~
\boldsymbol{\theta}=(\lambda,\rho),
\)
and the Mat\'ern covariance
\(
K_{\boldsymbol{\theta}}(d)
=
\lambda\frac{2^{1-\alpha}}{\Gamma(\alpha)}
\left(
\sqrt{2\alpha}\frac{d}{\rho}
\right)^\alpha
F_\alpha\left(
\sqrt{2\alpha}\frac{d}{\rho}
\right),~\boldsymbol{\theta}=(\lambda,\rho,\alpha),
\)
where $F_\alpha$ is the modified Bessel function of the second kind. Different covariance functions impose different assumptions on the regularity of $f_0$ \citep{kanagawa2018gaussianprocesseskernelmethods}. For the Mat\'ern covariance, $\alpha$ controls the smoothness, or differentiability, of GP sample paths, while the lengthscale $\rho$ controls the sensitivity of $f_0$ to changes in the inputs. Larger values of $\rho$ correspond to smoother estimates, while smaller values result in more ``wiggly'' estimates \citep{williams2006gaussian}.

For the nonparametric regression model in (\ref{mod:np}), GP regression treats the fixed target function $f_0$ as a realization from a GP and assumes $f_0\sim GP(0,K_{\boldsymbol{\theta}}).$ Given observations $(\mathbf{x}_i,y_i)$, $i=1,\ldots,n$, let $\boldsymbol{\Sigma}_{\boldsymbol{\theta}}=\left[
K_{\boldsymbol{\theta}}(\mathbf{x}_i,\mathbf{x}_j)
\right]_{i,j=1}^n$, so
\(
\mathbf{Y}
\sim
N\left(
\mathbf{0},
\boldsymbol{\Sigma}_{\boldsymbol{\theta}}+\sigma^2\mathbf{I}
\right).
\)
The best linear unbiased predictor (BLUP) of $f_0(\mathbf{x})$ is
\(
\widehat{f}_0(\mathbf{x})
=
\boldsymbol{\Sigma}_{1\boldsymbol{\theta}}
\left(
\boldsymbol{\Sigma}_{\boldsymbol{\theta}}
+\sigma^2\mathbf{I}
\right)^{-1}
\mathbf{Y},
\)
where $[\boldsymbol{\Sigma}_{1\boldsymbol{\theta}}]_{1i}
=
K_{\boldsymbol{\theta}}(\mathbf{x},\mathbf{x}_i)$ and $\boldsymbol{\Sigma}_{1\boldsymbol{\theta}}$ is $1\times n$. The covariance parameters $\boldsymbol{\theta}$ and $\sigma^2$ can be estimated using restricted maximum likelihood (REML) \citep{corbeil1976restricted}. 

The GP regression model naturally extends to a semiparametric model,
\begin{equation}\label{eqn:semi}
y_i
=
\mathbf{z}_i^\top\boldsymbol{\beta}
+
f_0(\mathbf{x}_i)
+
\epsilon_i,~
f_0\sim GP(0,K_{\boldsymbol{\theta}}),
\end{equation}
where $\mathbf{z}_i$ is a vector of covariates whose effects are modeled linearly and typically includes an intercept, and $\boldsymbol{\beta}$ is the corresponding vector of regression coefficients.

\section{Pairwise Additive GPs (PAGPs)}

In this section, we introduce our Pairwise Additive GP (PAGP) model.

\subsection{Additive GPs}

For a general additive GP model with $L$ terms,
\(
y_i = \mathbf{z}_i^T\boldsymbol{\beta}+\sum_{\ell=1}^L f_\ell(\mathbf{x}_i)+\epsilon_i,~
f_\ell\sim GP(0, K_\ell),
\)
where the $f_\ell$ are independent, the model can be equivalently written in the standard semiparametric GP form in \Cref{eqn:semi} as
\(
y_i=\mathbf{z}_i^T\boldsymbol{\beta}+f_0(\mathbf{x}_i)+\epsilon_i,~
f_0\sim GP(0, K),
\)
where $ K=\sum_{\ell=1}^L K_\ell.$
In other words, the covariance function of the GP prior for $f_0=\sum_{\ell=1}^L f_\ell$ is the sum of the covariance functions of the individual components \citep{durrande2012additive,duvenaud2011additivegaussianprocesses}. Supplemental Section \ref{s:computing} describes how to compute this sum efficiently for certain PAGP priors.

\subsection{PAGP Hypothesis Tests}\label{s:varimp_tests}

We now introduce the hypothesis testing framework based on the PAGP. Throughout this section, we use the squared-exponential base kernel.
We define $K_j(\mathbf{x},\mathbf{x}')=K_\theta(x_j,x_j')$ and $K_{jk}(\mathbf{x},\mathbf{x}')=K_\theta(\mathbf{x}_{jk},\mathbf{x}_{jk}')$, where $\mathbf{x}_{jk}=(x_j,x_k)^T$. Thus, $K_j$ models the main effect of $x_j$, while $K_{jk}$ models the pairwise interaction between $x_j$ and $x_k$.

As an example, we first consider testing
\(
H_0\text{: }x_1\text{ is not associated with }y.
\)
To formulate this test, we write the semiparametric PAGP model as
\begin{align}
\begin{split}\label{mod:x1}
y_i
&=
\mathbf{z}_i^T\boldsymbol{\beta}
+f_1^*(\mathbf{x}_i)
+f_{-1}^*(\mathbf{x}_i)
+\epsilon_i,~
\epsilon_i\sim\text{ i.i.d. }N(0,\sigma^2),
\end{split}
\end{align}
where
\(
f_1^*
=
f_1(x_1)+\sum_{j=2}^p f_{1j}(x_1,x_j)
\)
contains all terms involving $x_1$, and
\(
f_{-1}^*
=
\sum_{j=2}^p f_j(x_j)
+
\sum_{j>1,j<k}f_{jk}(x_j,x_k)
\)
contains all remaining main and interaction effects. Therefore, $H_0$ is equivalent to testing whether the component $f_1^*$ can be removed from the model.

Following the garrote-kernel framework of \cite{maity2011powerful}, we introduce an auxiliary parameter $\delta$ that controls the contribution of $f_1^*$. Specifically, we assign
\(
f_1^*+f_{-1}^*\sim GP(0,K_g),~
K_g=\delta K_1^*+\lambda K_{-1}^*,
\)
with
\(
K_1^*
=
\frac{1}{L}\left(K_1+\sum_{j=2}^pK_{1j}\right),~
K_{-1}^*
=
\frac{1}{L}\left(
\sum_{j=2}^pK_j+\sum_{j>1,j<k}K_{jk}
\right),
\)
and $L=p+{p\choose2}$ is the total number of main and pairwise interaction terms. Setting $\delta=0$ removes all kernel terms involving $x_1$, and hence removes $f_1^*$ from the model. Therefore, the hypothesis of no association between $x_1$ and $y$ can be written as
\(
H_0:\delta=0.
\)
Dividing the kernels by $L$ ensures that the scale of $\lambda$ is unaffected by $p$. Main and interaction effects are orthogonalized using a standard scheme to aid identifiability, particularly for tests of interaction (Supplemental Section \ref{s:orth}).

The GP model has an equivalent linear mixed model representation,
\begin{align}
\begin{split}\label{mod:pair_lmm}
\mathbf{Y}
&=
\mathbf{Z}^T\boldsymbol{\beta}
+\boldsymbol{\phi}_1^*
+\boldsymbol{\phi}_{-1}^*
+\boldsymbol{\epsilon},\\
\boldsymbol{\phi}_1^*
&\sim N(\mathbf{0},\delta\boldsymbol{\Sigma}_1^*),~
\boldsymbol{\phi}_{-1}^*
\sim N(\mathbf{0},\lambda\boldsymbol{\Sigma}_{-1}^*),~
\boldsymbol{\epsilon}\sim N(\mathbf{0},\sigma^2\mathbf{I}),
\end{split}
\end{align}
where $[\boldsymbol{\Sigma}_1^*]_{ij}=K_1^*(\mathbf{x}_i,\mathbf{x}_j)$ and $[\boldsymbol{\Sigma}_{-1}^*]_{ij}=K_{-1}^*(\mathbf{x}_i,\mathbf{x}_j)$. Thus
\(
\mathbf{Y}
\sim
N\left(
\mathbf{Z}^T\boldsymbol{\beta},
\delta\boldsymbol{\Sigma}_1^*
+\lambda\boldsymbol{\Sigma}_{-1}^*
+\sigma^2\mathbf{I}
\right)
\), $\delta$ is a variance component, and testing whether $x_1$ is associated with $y$ reduces to a variance component test of $H_0:\delta=0$.

Following \cite{zhang2003hypothesis,maity2011powerful,liu2022cross}, we use a variance component score test. Let
\(
\mathbf{V}_{\boldsymbol{\theta}}
=
\delta\boldsymbol{\Sigma}_1^*
+\lambda\boldsymbol{\Sigma}_{-1}^*
+\sigma^2\mathbf{I},~
\boldsymbol{\theta}=(\rho,\lambda,\delta,\sigma^2)^T.
\)
The parameters are estimated using REML under $H_0:\delta=0$. Denote the resulting estimates by $\widetilde{\boldsymbol{\theta}}$ and $\widetilde{\boldsymbol{\beta}}$, and let $\mathbf{V}_0=\mathbf{V}_{\widetilde{\boldsymbol{\theta}}}$. The score-based test statistic of \cite{zhang2003hypothesis} is
\(
T_n
=
(\mathbf{Y}-\mathbf{Z}^T\widetilde{\boldsymbol{\beta}})^T
\mathbf{V}_0^{-1}
\left(
\frac{\partial\mathbf{V}_{\boldsymbol{\theta}}}{\partial\delta}
\right)
\mathbf{V}_0^{-1}
(\mathbf{Y}-\mathbf{Z}^T\widetilde{\boldsymbol{\beta}}).
\)
Because $\partial\mathbf{V}_{\boldsymbol{\theta}}/\partial\delta
=\boldsymbol{\Sigma}_1^*$, this simplifies to
\[
T_n
=
(\mathbf{Y}-\mathbf{Z}^T\widetilde{\boldsymbol{\beta}})^T
\mathbf{V}_0^{-1}
\boldsymbol{\Sigma}_1^*
\mathbf{V}_0^{-1}
(\mathbf{Y}-\mathbf{Z}^T\widetilde{\boldsymbol{\beta}}).
\]
Following \cite{zhang2003hypothesis,maity2011powerful}, the null distribution of $T_n$ is approximated by an $m\chi_d^2$ distribution, where $m$ and $d$ are obtained by matching the first two moments using the Satterthwaite approximation and the efficient information for $\delta$; details are provided in Supplemental Section \ref{s:tn_null}. 

Tests for a group of predictors' joint association with $y$, specific pairwise interactions, and any pairwise interaction are constructed using the same framework by regrouping the corresponding PAGP terms (Supplemental Section \ref{s:extensions}). In practice, we estimate the lengthscale $\rho$ under $H_1$ rather than $H_0$ (Supplemental Section \ref{s:rho}). We discuss tradeoffs of this parameterization in Supplemental Section \ref{s:assns}.

\subsection{Interpreting results}\label{s:interp}

Rejecting the null hypothesis of no association between $x_j$ and $y$ does not give any information on the direction, shape, or type (e.g., main or interaction effect) of association. In addition, the reliability of our hypothesis tests may be affected by issues such as highly influential points and extrapolation difficult to identify in high-dimensional data.

To allow meaningful interpretation and check for influential points and other issues, we identify significant main and interaction effects using our proposed method and visualize them using a generalized additive model (GAM), which nonparametrically estimates smooth terms using smoothing splines \citep{wood2017generalized}. A properly-specified GAM has better predictive accuracy than our method, which is designed for hypothesis testing rather than prediction. Furthermore, a GAM decomposes $f_0$ into main and interaction effects with good uncertainty quantification. The primary weakness of GAMs is that it is difficult to choose which interaction effects to include, and including a large number of interactions leads to a cumbersome design matrix and computationally complex model. Our hypothesis tests solve these challenges by allowing rigorous selection of relevant interactions.

\subsection{Theoretical Justification}\label{s:theory}

We provide two theoretical results that explain the benefits of imposing additive structure on a GP. First, we derive posterior contraction rates showing that the effective dimension governing estimation is the largest dimension of the individual additive components, rather than the dimension of the full predictor vector. Thus, the contraction rate of PAGP depends on component dimensions of at most two, alleviating the curse of dimensionality. Second, we compare the sizes of the function spaces underlying hypothesis tests using metric entropy. This comparison shows that the PAGP substantially restricts the alternative space relative to a standard GP, providing heuristic justification for its improved hypothesis testing power. Note that these two results justify the benefit of the additive structure underlying PAGP, but do not provide a direct theoretical guarantee on the power of the proposed tests, as such power guarantees are generally intractable for nonparametric GP-based hypothesis tests and are also missing in the existing literature to the best of our knowledge.

\subsubsection{Posterior contraction rates}\label{s:rates}
We first study posterior contraction rates for a general additive GP. Suppose the target function can be written as
\(
f_0(\mathbf{x})=\sum_{l=1}^L g_l(\mathbf{x}_l),
\)
where $\mathbf{x}_l\in\mathcal{X}_l\subset\mathbb{R}^{p_l}$ is a subvector of $\mathbf{x}\in\mathcal{X}\subset\mathbb{R}^p$, and $\mathcal{X}$ and $\mathcal{X}_l$ are compact. Thus, although $f_0$ is a function of the full $p$-dimensional predictor, each component depends on only $p_l$ predictors. We model $f_0$ using the additive GP
\(
W(\mathbf{x})=\sum_{l=1}^L W_l(\mathbf{x}_l),
\)
where the $W_l$ are independent GPs with mean $0$. Let $\Pi$ denote the induced prior distribution of $W$ and let $\Pi_n(\cdot)=\Pi(\cdot\mid\mathbf{Y})$ denote its posterior distribution.

\begin{definition}
We say that the posterior contracts around $f_0$ at rate $\epsilon_n$ if
\[
\Pi_n\left(W:\|W-f_0\|>M_n\epsilon_n\right)\rightarrow 0
\]
in $P_{f_0}$-probability as $n\rightarrow\infty$, for every sequence $M_n\rightarrow \infty$.
\end{definition}
Let $C^\alpha(\mathcal{X})$ and $H^\alpha(\mathcal{X})$ denote the H\"older and Sobolev spaces, respectively~\citep{vandervaart11a}. The empirical norm is $\|g\|^2_n = \frac{1}{n}\sum_{i=1}^n g(\mathbf{x}_i)^2$. The following theorem gives the posterior contraction rate for the additive GP.

\begin{theorem}\label{thm:contract}
Assume that $g_l\in C^{\beta_l}(\mathcal{X}_l)\cap H^{\beta_l}(\mathcal{X}_l)$ and that the design points are fixed. Assume that $W_l$ has a Mat\'ern covariance with smoothness parameter $\alpha_l>0$. Then in the empirical norm, the posterior distribution of $W$ contracts around $f_0$ at rate
\[
\epsilon_n=n^{-\eta}, \eta
=
\min_{l=1,\ldots,L}
\left\{
\frac{\min\{\alpha_l,\beta_l\}}{2\alpha_l+p_l}
\right\}.
\] 
\end{theorem}

The proof is provided in Supplemental Section \ref{s:contract_proof}. Theorem \ref{thm:contract} shows that the contraction rate depends on the dimensions $p_l$ of the individual additive components rather than directly on the full dimension $p$. In particular, if
$\alpha_1=\cdots=\alpha_L=\beta_1=\cdots=\beta_L$, the contraction rate is the minimax rate for nonparametric regression with predictor dimension $\max_l p_l$ \citep{vandervaart11a}. For the PAGP, each component involves either one or two predictors, so $\max_l p_l=2$ regardless of $p$. Thus, imposing the PAGP structure alleviates the curse of dimensionality when the true regression function contains only main and pairwise interaction effects. The posterior contraction rate implies convergence of the posterior mean to $f_0$ in the empirical norm at the same rate \citep{vandervaart11a}.

We use Mat\'ern base kernels for this theoretical result, whereas the implementation uses squared-exponential base kernels because they allow efficient computation of the PAGP kernel (Supplemental Section \ref{s:computing}). Establishing optimal contraction rates for GPs with squared-exponential kernels is substantially more involved \citep{van2009adaptive}; we therefore use Mat\'ern kernels to characterize the effect of additive structure on the contraction rate.

\subsubsection{Size of parameter spaces under the null and alternative}\label{s:size}

We next compare the sizes of the function spaces associated with additive and standard GPs. This comparison provides heuristic justification for why restricting the alternative hypothesis to the PAGP structure can improve hypothesis testing power. Let $\Theta_0$ and $\Theta_1$ denote the parameter spaces under the null and alternative hypotheses, respectively, where each element of $\Theta_j$ is a regression function $f$. We quantify the size of these spaces using metric entropy, a standard measure of the size of a function space \citep{lorentz1962metric}. Specifically, the metric entropy $H_\epsilon(\Theta)$ is the logarithm of the minimum number of balls of radius $\epsilon$ required to cover $\Theta$.

The following result compares the metric entropy of an additive function space with that of an unrestricted function space having the same smoothness.

\begin{theorem}\label{fct:size}
Let $\Theta$ be a space of functions $f$ on a compact domain $\mathcal{X}\subset\mathbb{R}^p$. Up to a constant,
\[
H_\epsilon(\Theta)\lesssim
\begin{cases}
  \epsilon^{-p/\beta},
& f\in C^\beta(\mathcal{X})\cap H^\beta(\mathcal{X}),\\
\epsilon^{-\max_l p_l/\beta},
& f=\sum_{l=1}^L g_l,\quad
g_l\in C^\beta(\mathcal{X}_l)\cap H^\beta(\mathcal{X}_l),
\quad \mathcal{X}_l\subset\mathbb{R}^{p_l}.
\end{cases}
\]
\end{theorem}

The proof is provided in Supplemental Section \ref{s:size_proof}. Theorem \ref{fct:size} shows that the size of an unrestricted function space depends on the full predictor dimension $p$, whereas the size of an additive function space depends only on the largest component dimension $d=\max_l p_l$. For the PAGP, $d=2$ regardless of $p$. Thus, the function space allowed by PAGP is substantially smaller than that allowed by a standard GP as $p$ increases.

This distinction is particularly relevant under the alternative hypothesis. A standard GP allows the alternative to contain general $p$-dimensional functions, whereas the PAGP restricts it to sums of functions involving at most two predictors. This provides heuristic justification for the greater power of PAGP than a standard GP when the additive structure is appropriate.

\section{Simulation Studies}

We conduct simulation studies to evaluate the empirical performance of PAGP hypothesis tests. We first examine how hypothesis testing performance changes with predictor dimension, and then evaluate Type I error and power for testing individual predictor associations, interaction effects, and joint associations of multiple predictors. We additionally report prediction accuracy as a secondary measure of model performance, although prediction is not the primary objective of PAGP.

We compare PAGP with BKMR, GKM, a pairwise-additive version of GKM, GAM-based hypothesis tests implemented using \texttt{mgcv} \citep{wood2013p}, and ordinary least-squares regression (OLS). The pairwise-additive GKM uses the same testing framework as GKM but replaces the standard GP with a sum of two-dimensional GP components, without the orthogonalization used in PAGP. This comparison helps separate the effect of restricting the interaction order from the effect of our testing construction. Nominal Type I error rate was 0.10 for frequentist methods. BKMR uses posterior inclusion probabilities (PIPs) to indicate variable importance, and used PIP cutoffs of 0.50 and 0.90 to reject $H_0$.


\subsection{Simulation Setup}

For all simulations, predictors were generated independently as $x_j\sim \mathrm{Unif}(0,1)$, and
\[
y_i=f_0(\mathbf{x}_i)+\epsilon_i,~
\epsilon_i\overset{\mathrm{i.i.d.}}{\sim}N(0,\sigma^2),
\]
where $\sigma^2$ was chosen so that the signal-to-noise ratio
$\mathrm{SNR}
=
\frac{\operatorname{Var}\{f_0(\mathbf{x}_i)\}}{\sigma^2}=1$. We considered the following choices of $f_0$:
\begin{align}
    f_0(\mathbf{x}) &= \cos(2\pi x_1) + \cos(2\pi x_2) + 15(x_3 - 0.5)(x_4 - 0.5),\label{sim:cos} \\
    f_0(\mathbf{x}) &= \sin(2\pi x_1) + \sin(2\pi x_2) + 15(x_3 - 0.5)(x_4 - 0.5),\label{sim:sin}\\
    f_0(\mathbf{x}) &= 15(x_1 - 0.5)(x_2 - 0.5) + \sin(2\pi x_3) + \sin(2\pi x_4),\label{sim:int} \\
    f_0(\mathbf{x}) &= \sin(2\pi x_1) + \sin(2\pi x_2) + \sin(2\pi x_6) + \sin(2\pi x_7).\label{sim:joint}
\end{align}

In (\ref{sim:cos}), $x_1$ is associated with $y$ but $\operatorname{Cov}(x_1,y)=0$, whereas in (\ref{sim:sin}), $x_1$ is associated with $y$ and $\operatorname{Cov}(x_1,y)\neq 0$. These two scenarios therefore compare performance when the association cannot and can, respectively, be detected by a test of linear association.  Scenario (\ref{sim:int}) represents a pure interaction: $x_1$ affects $y$ only through its interaction with $x_2$ without any marginal effect. In fact, because $x_1, x_2$ are independent and uniformly distributed on $[0,1]$, the marginal effect of $x_1$ is $f_{mar,1}(x_1)=\int_{0}^1(x_1 - 0.5)(x_2 - 0.5)dx_2=0$. Scenario (\ref{sim:joint}) was used to test the joint association of $x_1,\ldots,x_5$ with $y$. We used sample sizes $n=100$ and $n=300$. Each simulation setting was repeated 2000 times.

\subsection{Simulation results}\label{s:sim_results}

We first examine how hypothesis testing performance changes with predictor dimension. Figure \ref{fig:dim_n100_power} shows the power and Type I error for testing the association between $x_1$ and $y$ under (\ref{sim:sin}) as $p$ increases. The power of GKM, which uses a standard GP and allows interactions of arbitrary order, decreases sharply as $p$ increases. In contrast, the power of PAGP and pairwise GKM decreases more slowly, consistent with the benefit of restricting the regression function to main and pairwise interaction effects. Type I error remains relatively stable as the dimension increases.

\begin{figure}
    \centering
    \includegraphics[width=0.9\linewidth]{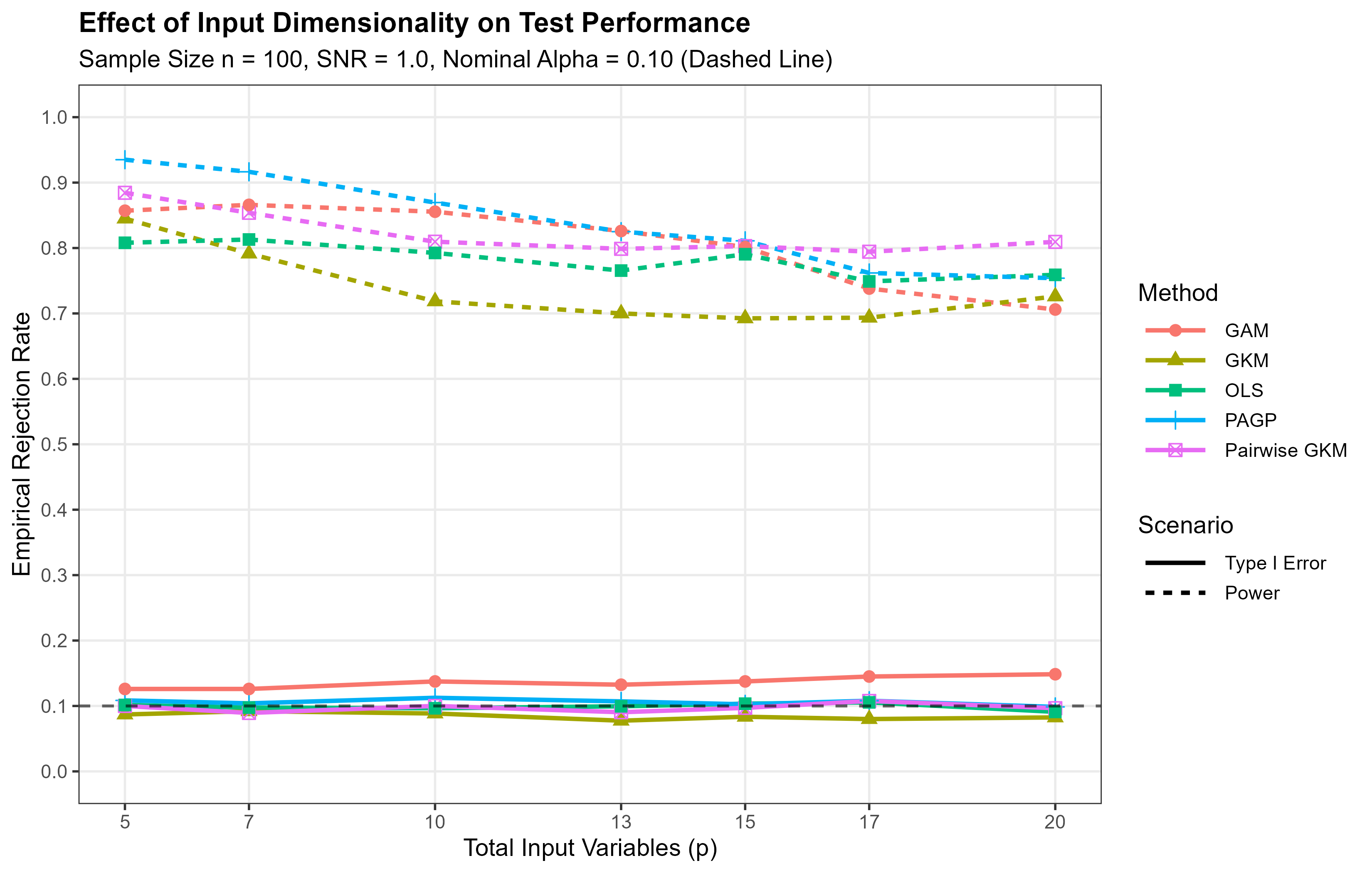}
    \caption{Impact of predictor dimension on power and Type I error for testing the association between $x_1$ and $y$ under (\ref{sim:sin}) with $n=100$.}
    \label{fig:dim_n100_power}
\end{figure}

We next compare methods for testing the association between an individual predictor and the outcome. Table \ref{tab:var_imp_sims} reports results for testing the association between $x_1$ and $y$ under the cosine, sine, and pure interaction scenarios.

\begin{table}[!htbp]
\centering
\small
\caption{Simulation results for testing the association between $x_1$ and $y$.}
\label{tab:var_imp_sims}
\begin{spacing}{0.7}
\begin{tabular}{lllrrr}
\toprule
\footnotesize
True $f_0$ & $n$ & Method & Type I Error & Power & MSE\\
\midrule
\multirow{14}{*}{Cosine (\ref{sim:cos})}
& \multirow{7}{*}{100}
& BKMR (PIP $>$ 0.5) & 0.094 & 0.656 & 0.264 \\
&& BKMR (PIP $>$ 0.9) & 0.006 & 0.292 & 0.264 \\
&& GAM & 0.144 & 0.792 & 0.515 \\
&& GKM & 0.091 & 0.116 & 0.424 \\
&& OLS & 0.098 & 0.121 & 0.643 \\
&& Pairwise GKM & 0.094 & 0.118 & 0.446 \\
&& PAGP & 0.125 & 0.625 & 0.386 \\
\cmidrule{2-6}
& \multirow{7}{*}{300}
& BKMR (PIP $>$ 0.5) & 0.012 & 0.999 & 0.062 \\
&& BKMR (PIP $>$ 0.9) & 0.002 & 0.988 & 0.062 \\
&& GAM & 0.112 & 1.000 & 0.370 \\
&& GKM & 0.100 & 0.112 & 0.183 \\
&& OLS & 0.108 & 0.099 & 0.548 \\
&& Pairwise GKM & 0.103 & 0.102 & 0.168 \\
&& PAGP & 0.100 & 0.970 & 0.169 \\
\midrule

\multirow{14}{*}{Sine (\ref{sim:sin})}
& \multirow{7}{*}{100}
& BKMR (PIP $>$ 0.5) & 0.123 & 0.883 & 0.234 \\
&& BKMR (PIP $>$ 0.9) & 0.007 & 0.550 & 0.234 \\
&& GAM & 0.134 & 0.862 & 0.543 \\
&& GKM & 0.081 & 0.742 & 0.363 \\
&& OLS & 0.099 & 0.810 & 0.510 \\
&& Pairwise GKM & 0.098 & 0.830 & 0.402 \\
&& PAGP & 0.102 & 0.878 & 0.350 \\
\cmidrule{2-6}
& \multirow{7}{*}{300}
& BKMR (PIP $>$ 0.5) & 0.008 & 0.994 & 0.113 \\
&& BKMR (PIP $>$ 0.9) & 0.001 & 0.972 & 0.113 \\
&& GAM & 0.112 & 1.000 & 0.377 \\
&& GKM & 0.078 & 1.000 & 0.205 \\
&& OLS & 0.100 & 0.998 & 0.423 \\
&& Pairwise GKM & 0.086 & 1.000 & 0.194 \\
&& PAGP & 0.085 & 0.999 & 0.188 \\
\midrule

\multirow{14}{*}{Pure interaction (\ref{sim:int})}
& \multirow{7}{*}{100}
& BKMR (PIP $>$ 0.5) & 0.306 & 0.992 & 0.237 \\
&& BKMR (PIP $>$ 0.9) & 0.021 & 0.922 & 0.237 \\
&& GAM & 0.134 & 0.204 & 0.548 \\
&& GKM & 0.104 & 0.531 & 0.369 \\
&& OLS & 0.100 & 0.144 & 0.512 \\
&& Pairwise GKM & 0.091 & 0.437 & 0.407 \\
&& PAGP & 0.110 & 0.951 & 0.352 \\
\cmidrule{2-6}
& \multirow{7}{*}{300}
& BKMR (PIP $>$ 0.5) & 0.072 & 1.000 & 0.112 \\
&& BKMR (PIP $>$ 0.9) & 0.005 & 1.000 & 0.112 \\
&& GAM & 0.108 & 0.181 & 0.375 \\
&& GKM & 0.098 & 0.665 & 0.203 \\
&& OLS & 0.101 & 0.154 & 0.420 \\
&& Pairwise GKM & 0.082 & 0.986 & 0.192 \\
&& PAGP & 0.109 & 0.999 & 0.187 \\
\bottomrule

\end{tabular}
\end{spacing}
\end{table}

The cosine and sine scenarios illustrate the distinction between nonlinear and linear association. Under (\ref{sim:cos}), $x_1$ is associated with $y$ but $\operatorname{Cov}(x_1,y)=0$. OLS, GKM, and Pairwise GKM therefore have little power, even as the sample size increases from $n=100$ to $n=300$. In contrast, PAGP detects the nonlinear association, with power increasing from 0.625 to 0.970. GAM and BKMR also detect the association, although GAM has somewhat elevated Type I error for $n=100$, while the Type I error of BKMR depends strongly on the PIP cutoff.

Under (\ref{sim:sin}), $x_1$ is associated with $y$ and $\operatorname{Cov}(x_1,y)\neq0$, so methods based on linear association perform substantially better. For $n=100$, PAGP has power 0.878, compared with 0.742 for GKM and 0.830 for Pairwise GKM, while maintaining Type I error close to the nominal level. The higher power of PAGP relative to GKM is consistent with the benefit of restricting the GP to main and pairwise interaction effects when higher-order interactions are absent. At $n=300$, all three methods have power close to one.

The pure interaction scenario (\ref{sim:int}) evaluates whether a method can detect the association between $x_1$ and $y$ when $x_1$ has no marginal main effect and affects $y$ through its interaction with $x_2$. PAGP has power 0.951 for $n=100$ and 0.999 for $n=300$. GAM and OLS have low power because their tests of the association of $x_1$ do not include its interaction with $x_2$. GKM has power 0.531 and 0.665, while Pairwise GKM has power 0.437 and 0.986. BKMR also has high power, although the PIP $>0.5$ rule has substantially elevated Type I error for $n=100$.

We next evaluate tests specifically targeting interaction effects. Table \ref{tab:int_sims} reports results for testing the specific interaction between $x_1$ and $x_2$ and for testing whether any pairwise interaction is present. Both tests use the pure interaction data-generating process (\ref{sim:int}).

\begin{table}[ht]
\centering
\caption{Simulation results for testing interaction effects under (\ref{sim:int}).}
\label{tab:int_sims}
\begin{tabular}{lllrr}
\toprule
Test & $n$ & Method & Type I Error & Power\\
\midrule
\multirow{6}{*}{$x_1,x_2$ interaction}
& \multirow{3}{*}{100}
& GAM & 0.166 & 1.000 \\
&& OLS & 0.106 & 1.000 \\
&& PAGP & 0.104 & 0.985 \\
\cmidrule{2-5}
& \multirow{3}{*}{300}
& GAM & 0.129 & 1.000 \\
&& OLS & 0.097 & 1.000 \\
&& PAGP & 0.098 & 1.000 \\
\midrule
\multirow{2}{*}{Any interaction}
& 100 & PAGP & 0.127 & 0.478 \\
& 300 & PAGP & 0.112 & 0.989 \\
\bottomrule
\end{tabular}
\end{table}

For the specific $x_1,x_2$ interaction, PAGP has power 0.985 at $n=100$ and 1.000 at $n=300$, with nominal Type I error rate. GAM and OLS also have power one because the corresponding interaction terms are explicitly included in these models; however, GAM has elevated Type I error, particularly at $n=100$. The global test for any interaction is more difficult at the smaller sample size, with PAGP attaining power 0.478 at $n=100$, but its power increases to 0.989 at $n=300$.

Finally, we evaluate joint testing of multiple predictors. Table \ref{tab:joint_sims} reports results for testing the joint association of $x_1,\ldots,x_5$ with $y$ under (\ref{sim:joint}).

\begin{table}[ht]
\centering
\caption{Simulation results for testing the joint association of $x_1,\ldots,x_5$ with $y$ under (\ref{sim:joint}).}
\label{tab:joint_sims}
\begin{tabular}{llrrr}
\toprule
$n$ & Method & Type I Error & Power & MSE\\
\midrule
\multirow{5}{*}{100}
& BKMR (PIP $>$ 0.5) & 0.844 & 1.000 & 0.297 \\
& BKMR (PIP $>$ 0.9) & 0.049 & 0.858 & 0.297 \\
& GAM & 0.220 & 0.996 & 0.222 \\
& OLS & 0.097 & 0.978 & 0.288 \\
& PAGP & 0.152 & 0.861 & 0.260 \\
\midrule
\multirow{5}{*}{300}
& BKMR (PIP $>$ 0.5) & 0.353 & 1.000 & 0.201 \\
& BKMR (PIP $>$ 0.9) & 0.021 & 1.000 & 0.201 \\
& GAM & 0.149 & 1.000 & 0.050 \\
& OLS & 0.098 & 1.000 & 0.224 \\
& PAGP & 0.121 & 0.992 & 0.148 \\
\bottomrule
\end{tabular}
\end{table}

For the joint association test, PAGP has power 0.861 at $n=100$ and 0.992 at $n=300$. Its Type I error decreases from 0.152 to 0.121 as the sample size increases. GAM has higher power but elevated Type I error at both sample sizes. BKMR is highly sensitive to the PIP cutoff: the 0.5 cutoff has Type I error of 0.844 and 0.353 for $n=100$ and $n=300$, respectively, whereas the 0.9 cutoff is conservative. OLS performs well in this particular scenario because $Cov(x_j,y)\neq 0$ for the sine functions.

Prediction MSE is reported as a secondary measure of model performance in Tables \ref{tab:var_imp_sims}, \ref{tab:int_sims}, and  \ref{tab:joint_sims}. 

Overall, these simulations show that PAGP provides reliable Type I error control while maintaining high power across nonlinear, interactive, and joint alternatives, with its advantage over standard GP-based tests becoming more pronounced as predictor dimension increases.

\section{Applications}
In this section, we analyze two data sets to demonstrate our method. First, we analyze associations between metals in private well water measured from 1998-2019 and small-area estimates of obesity prevalence among North Carolina (NC) census tracts. Second, we analyze associations between liver injury and blood concentrations of five metals using United States National Health and Nutrition Examination Survey (NHANES) data from 2017--2020. In Supplemental Section \ref{s:telomere}, we additionally analyze associations between mean leukocyte telomere length and 18 persistent organic pollutants (POPs) using NHANES 2001--2002 data.

For the PAGP analyses, we conduct hypothesis tests using our method and then interpret significant associations using GAMs, as described in Section \ref{s:interp}. Covariates are included as predictors in the PAGP rather than as linear adjustment terms, allowing nonlinear effects and interactions, except in the telomere analysis, where we match previous method comparisons. Our GAM main and interaction effect plots are based on the \texttt{gratia} R package \citep{Simpson2024}. These partial effects are centered to have mean zero and represent the additive contribution of each term to the expected outcome, rather than the expected outcome itself, which also depends on the values of the other predictors.

We compare PAGP with multiple linear regression, GAM, and GKM in the two main analyses. BKMR is additionally included in the supplemental telomere analysis but is omitted from the NC and liver analyses due to computational cost and difficulty achieving model convergence. All predictors are shifted and scaled to lie within [0,1]. For the NCWELL data, all methods included an additional spatially-correlated residual term to account for residual spatial autocorrelation in estimated obesity prevalence. Summary statistics for the data sets are in Supplemental Section \ref{s:supp_addl_results}.

\subsection{NCWELL Data}

In NC, 2.4 million people obtain their drinking water from private wells \citep{dieter2018estimated}. Unlike public water systems, private wells are not federally regulated, increasing the potential for exposure to harmful contaminants \citep{eaves2022analysis}. Inorganic arsenic (iAs), ranked second as of 2025 on the Agency for Toxic Substances and Disease Registry (ATSDR) priority list \citep{astdr_priority}, is of particular concern. Arsenic exposure is associated with a number of cancers and other negative health effects, including possible effects later in life from early-life exposure \citep{naujokas2013broad}. Inorganic arsenic has been linked to obesity in experimental, animal studies, although there is no consistent human evidence \citep{eick2020arsenic}. In NC, complicated social and geographic factors also contribute to health disparities and may confound exposure-response associations \citep{plescia2014reducing, barber2020place}.

We examine associations between 11 metals detected in private wells and estimated obesity prevalence among 1,693 NC census tracts with at least one private well measurement for each metal. We use metals detected in the well water of at least 70\% of census tracts for analysis: inorganic arsenic, lead, manganese, chromium, copper, zinc, iron, calcium, magnesium, sodium, and barium. We use well water tests conducted in 2019. All well water tests within a census tract were averaged using 2010 census tract boundaries. Well measurements of metals are from the NCWELL database which is described in detail in \cite{eaves2022analysis}. Small area estimates of obesity prevalence  in 2019 are created by the Centers for Disease Control Population Level Analysis and Community Estimates (CDC PLACES) \citep{greenlund2022places, zhang2014multilevel}. Obesity is calculated from self-reported height and weight, and we use estimated crude prevalence among adults (at least 18 years old). Tract-level covariates obtained from the American Community Survey \citep{acs2019_5yr} included five-year estimates of percent Black, percent Hispanic, median tract income, median tract age, and tract population. Tract centroid latitude and longitude were also included as covariates. Estimated percent of tract residents using private well water used results for 2020 from \citep{murray2025machine}. The small area estimates of obesity prevalence are influenced heavily by demographic variables, but include spatially-structured, non-demographic variation which may allow us to detect metals' effects \citep{kong2020use, zhang2014multilevel}. We include spatially-correlated residuals to account for the spatial structure of non-demographic variation in the outcome data. The exposure variables were the tract-level average measured well concentration in parts per billion for each metal. All arsenic measured in wells was inorganic.

Using PAGP, the metals were jointly associated with obesity ($p = 0.001$), but using generalized least squares (GLS), fit using the \texttt{GpGp} R package \citep{GpGp}, metals were not jointly significant ($p=0.717$). Results of the individual association tests for metals are displayed in Table \ref{tab:NCWELL}.
\begin{table}[!htbp]
\centering
\caption{P-values for tests of association between the average concentration (ppb) of metals in private wells in each NC census tract and estimated obesity prevalence, for 1,693 tracts and 11 metals detected in at least 70\% of tracts, adjusted for latitude and longitude of the tract centroid, estimated percent of tract residents using private well water, percent Black, percent Hispanic, median tract income, median tract age, and tract population. Because our focus is on hypothesis testing in nonparametric models, we report p-values rather than effect estimates and confidence intervals. Metals were natural-log transformed.}
\label{tab:NCWELL}
\begin{tabular}{lrrrr}
\toprule
 Variable & PAGP & GLS& GKM & GAM \\
\midrule
Inorganic Arsenic& 0.02& 0.252& 0.018& 0.023\\
Lead&  0.05& 0.290& 0.069& 0.472\\
Manganese& 0.13& 0.851& 0.074& 0.231\\
Chromium& 0.34& 0.374& 0.336& 0.160\\
Copper& 0.14& 0.218& 0.089& 0.144\\
Zinc& 0.11& 0.239& 0.126& 0.643\\
Iron& 0.86& 0.666& 0.798& 0.539\\
Calcium& 0.47& 0.694& 0.555& 0.637\\
Magnesium& 0.36& 0.875& 0.418& 0.552\\
Sodium& 0.39& 0.519& 0.672& 0.881\\
Barium& 0.46& 0.588& 0.373& 0.653\\
\bottomrule
\end{tabular}
\end{table}

PAGP identified significant associations with obesity prevalence for inorganic arsenic and lead at the 0.05 significance level. The PAGP test for any interaction found evidence of interactive effects ($p < 0.001$). The full list of significant ($p < 0.05$) interaction effects is displayed in the Supplemental Section \ref{s:ncwell_addl}. Because covariates were included as predictors in PAGP, the interaction tests also allowed for metal-covariate interactions. The significant interactions involving metals were between arsenic and lead ($p=0.044$), and arsenic and median income ($p=0.035$). GLS found no significant associations between metals and obesity, likely due to the highly significant and nonlinear effects of demographic variables on the small area estimates of obesity prevalence (Figure \ref{fig:ncwell_demo}), reducing power in the misspecified linear model. While our variable significance results were similar to the GAM, the GAM cannot easily identify relevant interactions.

We next use a GAM including main and interaction effects identified by PAGP to characterize these effects. The estimated main effects for arsenic and lead from the GAM are displayed in Figure \ref{fig:NCWELL_gam_main}. Arsenic has an increasing association with obesity prevalence in our data. Experimental evidence has shown that inorganic arsenic exposure can cause obesity in animals \citep{eick2020arsenic, tinkov2021adipotropic}. Despite this, observational studies of human health effects have not found consistent evidence of association between arsenic exposure and obesity \citep{eick2020arsenic}. Studies of human subjects may be confounded by exposure to organic arsenic and may suffer from reverse causation because adipose tissue stores arsenic \citep{eick2020arsenic}. Our ecologic study is limited by aggregated data and use of small area estimates, but our well measurements distinguish inorganic arsenic and there is no apparent mechanism for reverse causation in our analysis. This may explain why our results align more closely to experimental evidence of arsenic's causal effect on obesity than with previous observational studies.

\begin{figure}
    \centering
    
    \includegraphics[width=0.6\linewidth]{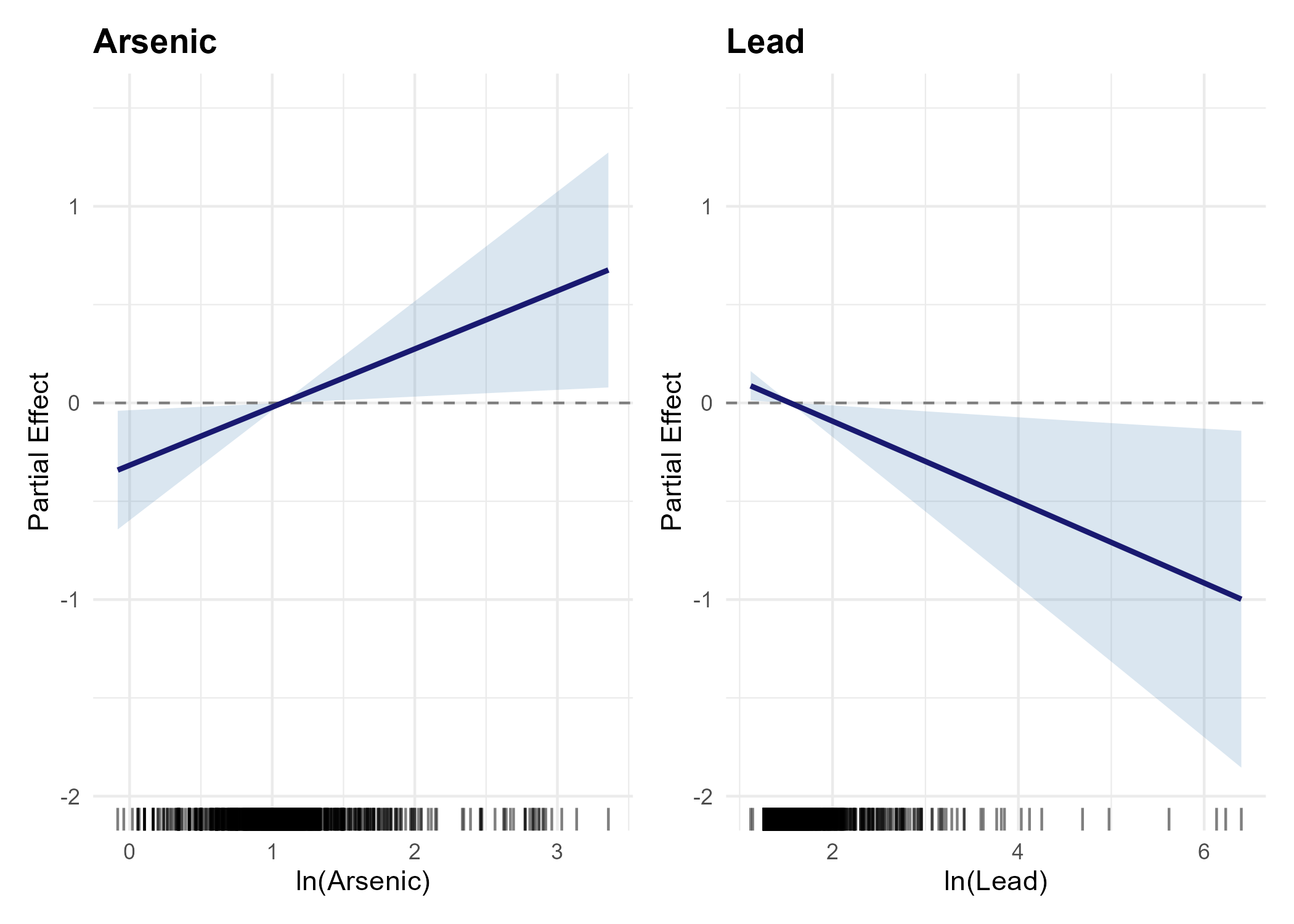}\\
    \caption{Selected partial main effects estimated by GAM smooth terms, using main and interaction effects identified by the PAGP hypothesis tests, for the NCWELL data. The partial effects characterize associations found between natural log-transformed metal concentrations (ppb) in private wells, averaged by census tract, and estimated obesity prevalence.}
        \label{fig:NCWELL_gam_main}

\end{figure}

Lead has a decreasing association with obesity prevalence in our data. Laboratory studies show causal effects of lead on increased obesity risk in animals \citep{tinkov2021adipotropic}, but a recent meta-analysis of observational studies \citep{zangiabadian2024association} found an association between lead exposure and decreased risk of obesity in observational studies of humans. Our study may be affected by confounding from similar sources as other observational studies of lead and obesity.

Finally, we examine the two significant interactions involving metals identified by PAGP. The estimated arsenic-median income interaction shows a stronger association between arsenic and obesity for higher levels of income (Figure \ref{fig:NCWELL_gam_as_income}); note that no observations have high values of both arsenic and income relative to the rest of the sample (Figure \ref{fig:as_income_2d}). We hypothesize that the higher background risk of obesity associated with lower income may mask risk associated with arsenic exposure. The interaction between arsenic and lead was not significant in the GAM. The GAM provides only one approach for characterizing their functional forms and may itself be limited by model specification and data availability. The interactions identified by PAGP should be viewed as candidates for further investigation, providing a starting point for more targeted scientific analyses and, where necessary, additional data collection.


\begin{figure}
    \centering
    
    \includegraphics[width=0.6\linewidth]{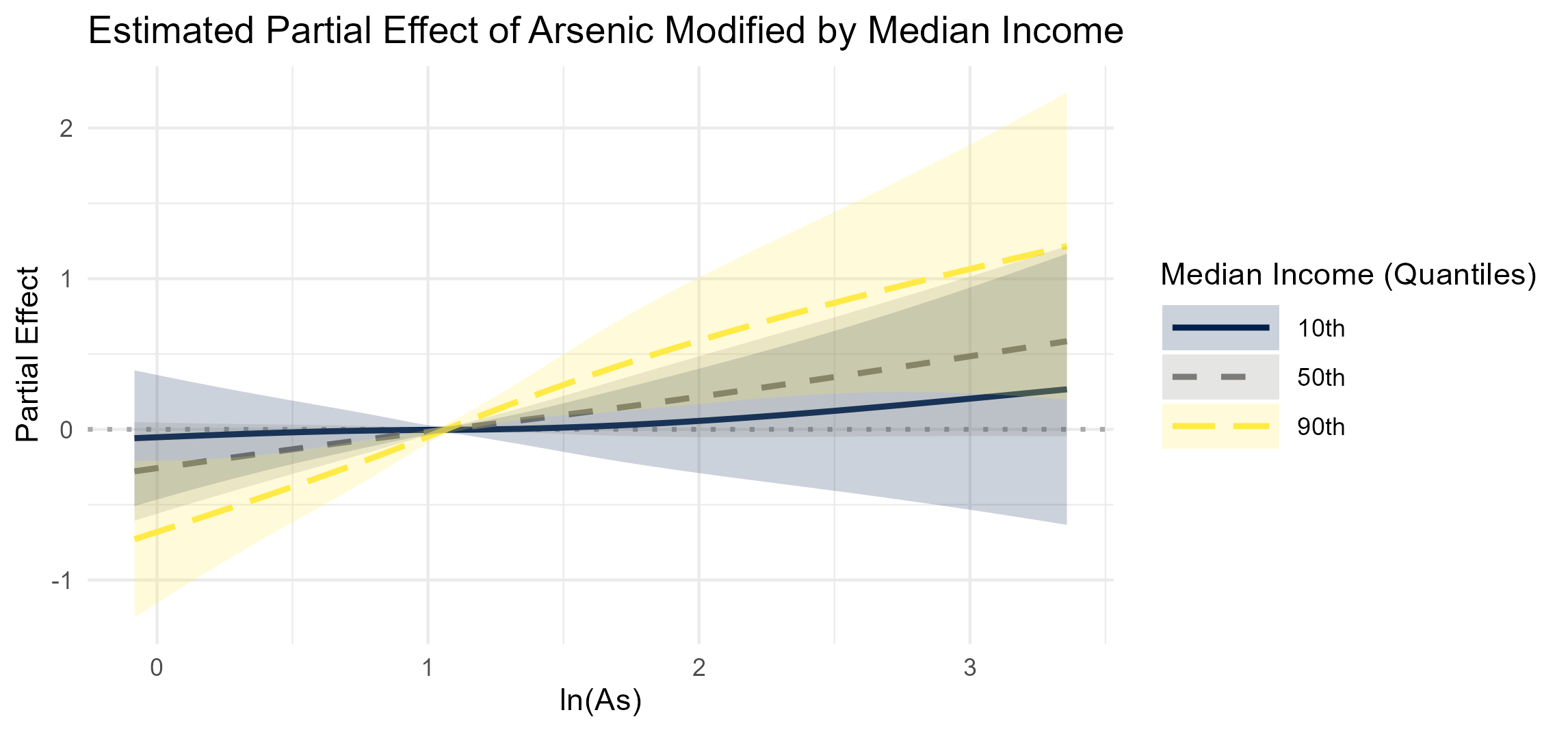}
    \caption{Additional visualization of the arsenic-income interaction effect estimated by the GAM tensor interaction smooth term, using main and interaction effects identified by PAGP hypothesis tests plus all covariates, for the NCWELL data. Note that there are no observations with both high arsenic and high income (Figure \ref{fig:as_income_2d}).}
        \label{fig:NCWELL_gam_as_income}

\end{figure}

\subsection{NHANES 2017-2020: Liver data}\label{s:liver}

Next, we study NHANES 2017--2020 data, focusing on associations between blood concentrations of five metals, mercury, manganese, lead, cadmium, and selenium, and alanine transaminase (ALT) concentration (U/L) in blood serum, similar to \cite{li2023associations}. Metals can accumulate in the liver and cause liver injury, and ALT is a commonly used biomarker of liver injury \citep{li2023associations}. After removing observations with missing values, the sample size is $n=2,507$. Covariates were alcohol use, diabetes, hypertension, sex, race, ethnicity, smoking status, age, body mass index (BMI), energy intake (kcal), cholesterol (mg/dL), high-density lipoprotein cholesterol (HDL-C, mg/dL), and poverty ratio. We natural log-transformed the metals and ALT to stabilize estimates and obtain approximately normal residuals.

Using PAGP and OLS tests, the metals were jointly associated with ALT ($p < 0.001$, both methods). Results of the individual association tests for metals are displayed in Table \ref{tab:liver_variables}.

\begin{table}[!htbp]
\centering
\caption{P-values for tests of association between the concentration of metals in blood and ALT concentration in blood (U/L) among 2,507 individuals in NHANES 2017--2020 data, adjusted for alcohol use, diabetes, hypertension, sex, race, ethnicity, smoking status, age, BMI, energy intake (kcal), cholesterol (mg/dL), HDL-C, and poverty ratio. Because our focus is on hypothesis testing in nonparametric models, we report p-values rather than effect estimates and confidence intervals. Metals and ALT were both natural log-transformed.}
\label{tab:liver_variables}
\begin{tabular}{lrrrr}
\toprule
Metal & PAGP & OLS & GKM & GAM\\
\midrule
Mercury ($\mu$g/L) & 0.191 & 0.051 & 0.114 & 0.117\\
Manganese ($\mu$g/L) & $<0.001$ & 0.044 & 0.002 & 0.115\\
Lead ($\mu$g/dL) & $<0.001$ & 0.078 & 0.001 & 0.058\\
Cadmium ($\mu$g/L) & $<0.001$ & 0.246 & 0.145 & 0.048\\
Selenium ($\mu$g/L) & $<0.001$ & $<0.001$ & $<0.001$ & 0.011\\
\bottomrule
\end{tabular}
\end{table}

PAGP identified significant associations with ALT for manganese, lead, cadmium, and selenium at the 0.05 significance level. The PAGP test for any interaction also found evidence of interactive effects ($p < 0.001$). We then tested individual interactions involving the significant variables. The significant interactions involving metals were manganese with HDL-C, lead with HDL-C, selenium with HDL-C, cadmium with HDL-C, manganese with male sex, and lead with hypertension. Age also interacted with several covariates. Full hypothesis test results are provided in Supplemental Section \ref{s:nhanes_addl}.

We next use GAMs to characterize the significant associations identified by PAGP. The estimated main effects are displayed in Figure \ref{fig:gam_alt_main}. Selenium has an increasing association with ALT, whereas lead has a decreasing association. Lead is known to cause liver injury, which releases ALT into the blood \citep{lakka2023potential}, so the inverse association may reflect confounding or sampling variation; the inverse association persisted using survey-weighted OLS. Using a larger and different NHANES sample, \cite{li2023associations} found each metal to be positively associated with ALT.

Cadmium has a nonlinear association with ALT. PAGP and GAM find evidence of an association between cadmium and ALT, whereas OLS and GKM do not. This result illustrates the potential benefit of testing nonlinear associations that may not be detected by tests of linear association. In contrast, OLS provides stronger evidence for an association between mercury and ALT than the more flexible methods. This association may be approximately linear but weak, in which case more flexible methods may have less power to detect it.

The association between manganese and ALT differs by sex. The GAM suggests little overall association between manganese and ALT but an increasing association among males (Figure \ref{fig:gam_alt_main}). Thus, the PAGP analysis extends the previous analysis of \cite{li2023associations} by identifying evidence that the association between manganese and liver health may differ between men and women.

\begin{figure}
    \centering
    
    \includegraphics[width=0.8\linewidth]{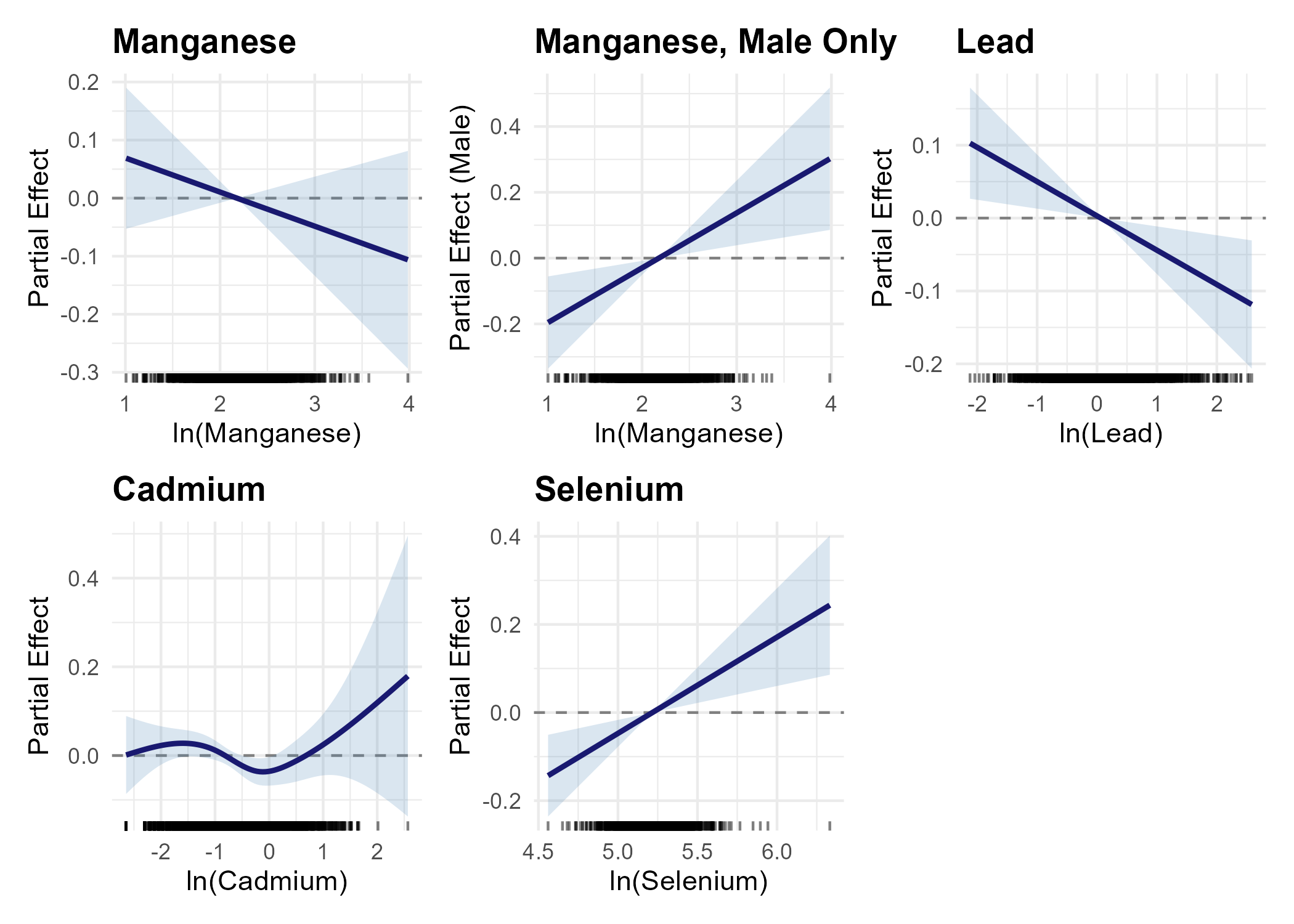}
    \caption{Selected partial univariate effects estimated by GAM smooth terms, using main and interaction effects identified by the PAGP hypothesis tests plus all covariates, for the NHANES 2017--2020 data. The partial effects characterize associations between natural log-transformed metals in blood and ln(ALT) in blood. The manganese-male interaction shows an increasing association between ln(Mn) and ln(ALT) among males.}
    \label{fig:gam_alt_main}
\end{figure}

Finally, we use GAMs to further examine the interactions involving metals identified by PAGP. The estimated manganese-HDL-C and cadmium-HDL-C interaction effects are displayed in Figure \ref{fig:gam_alt_int}. The GAM estimates of both interactions are highly influenced by individual observations in regions with limited data. As in the NCWELL analysis, the GAM provides an exploratory tool for characterizing the interactions identified by PAGP rather than a gold standard for validating them. The interactions identified by PAGP provide candidates for further scientific investigation and may motivate more targeted analyses or additional data collection.

\begin{figure}
    \centering
    
    \includegraphics[width=1\linewidth]{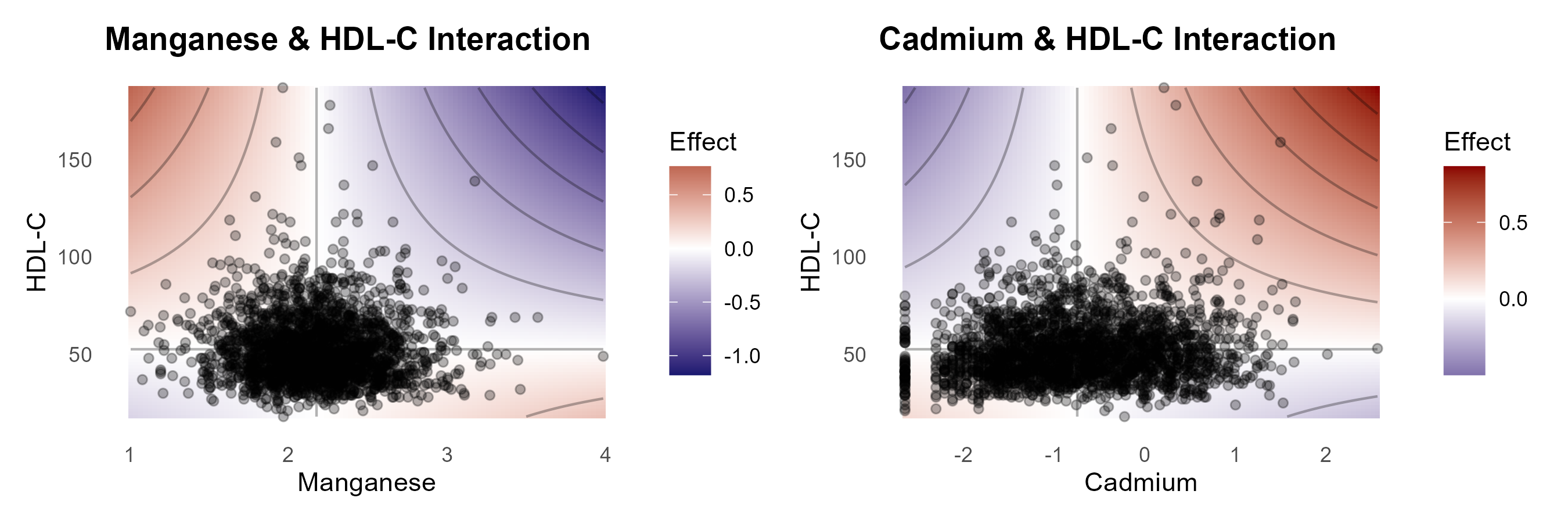}
    \caption{Selected partial interaction effects estimated by GAM smooth terms, using main and interaction effects identified by the PAGP hypothesis tests plus all covariates, for the NHANES 2017--2020 data. The partial effects characterize associations between natural log-transformed metals in blood and ln(ALT) in blood. The estimated manganese-HDL-C and cadmium-HDL-C interactions are highly influenced by individual observations in regions with limited data.}
        \label{fig:gam_alt_int}

\end{figure}

\section{Discussion}

We developed PAGP, a nonparametric framework for hypothesis testing of nonlinear main effects, pairwise interactions, joint associations, and the presence of any interaction. By restricting the GP to main and pairwise interaction effects, PAGP avoids unnecessary higher-order interactions and is less affected by increasing predictor dimension than standard GP-based tests. Simulations showed generally reliable Type I error control and high power across nonlinear and interactive alternatives, with particularly strong advantages as predictor dimension increased. Our theoretical results provide complementary justification for the additive structure: posterior contraction rates depend on the dimensions of the individual additive components rather than the full predictor dimension, while metric entropy results show that PAGP substantially restricts the alternative function space relative to a standard GP. Applications to NC private-well and NHANES data demonstrated the ability of PAGP to identify nonlinear and interactive associations and to generate scientifically relevant candidates for subsequent investigation.

Several limitations motivate future work. We currently use GAMs to characterize and visualize effects identified by PAGP; this secondary analysis is useful for interpretation and diagnostics but is itself model-dependent and should not be viewed as a gold standard for validating the PAGP findings. Rather, significant effects identified by PAGP can serve as a starting point for deeper scientific investigation and, when existing data provide limited information about their functional form, targeted additional data collection. In addition, the orthogonalization of main and interaction effects used here assumes independent predictors, motivating extensions to dependent predictor distributions \citep[e.g.,][]{hooker2007generalized}. Other important directions include scalable approximations to reduce the computational cost of exact GP inference, more flexible covariance parameterizations and testing procedures, and extensions to non-Gaussian outcomes.

\subsection*{Acknowledgements}
Support was provided through the National Institutes of Health (NIH) from the National Institute of Environmental Health Sciences (NIEHS), including grants P42ES031007 and R01ES029531. NW was supported by NIH grant T32ES007018-46. AQ was supported by NIH grant R37 AI029168. DL was supported by NIH grants R01 LM014407, P50 CA058223, and P42 ES031007.
\vspace*{-8pt}


\subsection*{Supplementary Materials}

Web Appendices and Tables referenced in Sections 1-5, and code for implementing the methods and conducting empirical studies are available with this paper at the Biometrics website on Oxford Academic. Code is also available on GitHub at \href{https://github.com/nbwiecha/Pairwise_Additive_GP_Tests}{nbwiecha/Pairwise\_Additive\_GP\_Tests}. \vspace*{-8pt}

\subsection*{Data availability}
The NCWELL data set is available publicly at \url{https://dataverse.unc.edu/dataset.xhtml?persistentId=doi:10.15139/S3/BDQG9O}, and the American Community Survey data is available at \url{https://data.census.gov/}. CDC PLACES data is available at \url{https://www.cdc.gov/places/tools/data-portal.html}. Private well use data is available on reasonable request to the corresponding author. NHANES data can be found at \url{https://wwwn.cdc.gov/nchs/nhanes/Default.aspx}.

\label{lastpage}
\newpage
 \bibliographystyle{rss} 
\bibliography{refs.bib}

@article{vandervaart11a,
  title={Information Rates of Nonparametric Gaussian Process Methods.},
  author={Van Der Vaart, Aad and Van Zanten, Harry},
  journal={Journal of Machine Learning Research},
  volume={12},
  number={6},
  year={2011}
}

@article{van2008rates,
  title={Rates of contraction of posterior distributions based on Gaussian process priors},
  author={Van der Vaart, Aad W and Van Zanten, J Harry},
  year={2008}
}

@article{lorentz1962metric,
  title={Metric entropy, widths, and superpositions of functions},
  author={Lorentz, GG},
  journal={The American Mathematical Monthly},
  volume={69},
  number={6},
  pages={469--485},
  year={1962},
  publisher={Taylor \& Francis}
}

@article{maity2011powerful,
  title={Powerful tests for detecting a gene effect in the presence of possible gene--gene interactions using garrote kernel machines},
  author={Maity, Arnab and Lin, Xihong},
  journal={Biometrics},
  volume={67},
  number={4},
  pages={1271--1284},
  year={2011},
  publisher={Oxford University Press}
}

@article{liu2022cross,
  title={A cross-validated ensemble approach to robust hypothesis testing of continuous nonlinear interactions: application to nutrition-environment studies},
  author={Liu, Jeremiah Zhe and Deng, Wenying and Lee, Jane and Lin, Pi-i Debby and Valeri, Linda and Christiani, David C and Bellinger, David C and Wright, Robert O and Mazumdar, Maitreyi M and Coull, Brent A},
  journal={Journal of the American Statistical Association},
  volume={117},
  number={538},
  pages={561--573},
  year={2022},
  publisher={Taylor \& Francis}
}

@article{bobb2015bayesian,
  title={Bayesian kernel machine regression for estimating the health effects of multi-pollutant mixtures},
  author={Bobb, Jennifer F and Valeri, Linda and Claus Henn, Birgit and Christiani, David C and Wright, Robert O and Mazumdar, Maitreyi and Godleski, John J and Coull, Brent A},
  journal={Biostatistics},
  volume={16},
  number={3},
  pages={493--508},
  year={2015},
  publisher={Oxford University Press}
}

@book{williams2006gaussian,
    author = {Rasmussen, Carl Edward and Williams, Christopher K. I.},
    title = "{Gaussian Processes for Machine Learning}",
    publisher = {The MIT Press},
    year = {2005},
    month = {11},
    isbn = {9780262256834},
    doi = {10.7551/mitpress/3206.001.0001},
    url = {https://doi.org/10.7551/mitpress/3206.001.0001},
}

@inproceedings{durrande2012additive,
  title={Additive covariance kernels for high-dimensional Gaussian process modeling},
  author={Durrande, Nicolas and Ginsbourger, David and Roustant, Olivier},
  booktitle={Annales de la Facult{\'e} des sciences de Toulouse: Math{\'e}matiques},
  volume={21},
  number={3},
  pages={481--499},
  year={2012}
}

@article{wiecha2026novel,
  title={When are novel methods for analyzing complex chemical mixtures in epidemiology beneficial?},
  author={Wiecha, Nate and Griffith, Emily and Reich, Brian J and Hoppin, Jane A},
  journal={Environmental Epidemiology},
  volume={10},
  number={1},
  pages={e456},
  year={2026},
  publisher={LWW}
}

@InProceedings{lu2022oak,
  title = 	 {Additive {G}aussian Processes Revisited},
  author =       {Lu, Xiaoyu and Boukouvalas, Alexis and Hensman, James},
  booktitle = 	 {Proceedings of the 39th International Conference on Machine Learning},
  pages = 	 {14358--14383},
  year = 	 {2022},
  editor = 	 {Chaudhuri, Kamalika and Jegelka, Stefanie and Song, Le and Szepesvari, Csaba and Niu, Gang and Sabato, Sivan},
  volume = 	 {162},
  series = 	 {Proceedings of Machine Learning Research},
  month = 	 {17--23 Jul},
  publisher =    {PMLR},
  url = 	 {https://proceedings.mlr.press/v162/lu22b.html}
}

@article {Morgensterne037860,
	author = {Morgenstern, Jason Denzil and Buajitti, Emmalin and O{\textquoteright}Neill, Meghan and Piggott, Thomas and Goel, Vivek and Fridman, Daniel and Kornas, Kathy and Rosella, Laura C},
	title = {Predicting population health with machine learning: a scoping review},
	volume = {10},
	number = {10},
	elocation-id = {e037860},
	year = {2020},
	doi = {10.1136/bmjopen-2020-037860},
	publisher = {British Medical Journal Publishing Group},
	issn = {2044-6055},
	URL = {https://bmjopen.bmj.com/content/10/10/e037860},
	eprint = {https://bmjopen.bmj.com/content/10/10/e037860.full.pdf},
	journal = {BMJ Open}
}

@article{mcgee2023bayesian,
  title={Bayesian multiple index models for environmental mixtures},
  author={McGee, Glen and Wilson, Ander and Webster, Thomas F and Coull, Brent A},
  journal={Biometrics},
  volume={79},
  number={1},
  pages={462--474},
  year={2023},
  publisher={Oxford University Press}
}

@book{molnar2025,
  title={Interpretable Machine Learning},
  subtitle={A Guide for Making Black Box Models Explainable},
  author={Christoph Molnar},
  year={2025},
  edition={3},
  isbn={978-3-911578-03-5},
  url={{ https://christophm.github.io/interpretable-ml-book}}
}

@book{hastie2009elements,
  title={The Elements of Statistical Learning},
  author={Hastie, Trevor and Tibshirani, Robert and Friedman, Jerome H and Friedman, Jerome H},
  volume={2},
  year={2009},
  publisher={Springer}
}

@book{gu2013smoothing,
  title={Smoothing spline ANOVA models},
  author={Gu, Chong and Gu, Chong},
  volume={297},
  year={2013},
  publisher={Springer}
}

@book{wood2017generalized,
  title={Generalized additive models: an introduction with R},
  author={Wood, Simon N},
  year={2017},
  publisher={chapman and hall/CRC}
}

@article{breiman2001random,
  title={Random forests},
  author={Breiman, Leo},
  journal={Machine learning},
  volume={45},
  number={1},
  pages={5--32},
  year={2001},
  publisher={Springer}
}

@article{zhang2003hypothesis,
  title={Hypothesis testing in semiparametric additive mixed models},
  author={Zhang, Daowen and Lin, Xihong},
  journal={Biostatistics},
  volume={4},
  number={1},
  pages={57--74},
  year={2003},
  publisher={Oxford University Press}
}

@article{corbeil1976restricted,
  title={Restricted maximum likelihood (REML) estimation of variance components in the mixed model},
  author={Corbeil, Robert R and Searle, Shayle R},
  journal={Technometrics},
  volume={18},
  number={1},
  pages={31--38},
  year={1976},
  publisher={Taylor \& Francis}
}

@article{kanagawa2018gaussianprocesseskernelmethods,
      title={Gaussian Processes and Kernel Methods: A Review on Connections and Equivalences}, 
      author={Motonobu Kanagawa and Philipp Hennig and Dino Sejdinovic and Bharath K Sriperumbudur},
      year={2018},
      eprint={1807.02582},
      archivePrefix={arXiv},
      primaryClass={stat.ML},
      url={https://arxiv.org/abs/1807.02582}, 
      journal={arXiv preprint arXiv:1807.02582}
}

@article{van2009adaptive,
  title={Adaptive Bayesian estimation using a Gaussian random field with inverse gamma bandwidth},
  author={Van der Vaart, Aad W and Van Zanten, J Harry},
  year={2009}
}

@article{wood2013p,
  title={On p-values for smooth components of an extended generalized additive model},
  author={Wood, Simon N},
  journal={Biometrika},
  volume={100},
  number={1},
  pages={221--228},
  year={2013},
  publisher={Oxford University Press}
}

@article{li2023associations,
  title={Associations of blood metals with liver function: Analysis of NHANES from 2011 to 2018},
  author={Li, Wenjie and Li, Xinyan and Su, Jing and Chen, Han and Zhao, Ping and Qian, Haisheng and Gao, Xin and Ye, Qiang and Zhang, Guoxin and Li, Xuan},
  journal={Chemosphere},
  volume={317},
  pages={137854},
  year={2023},
  publisher={Elsevier}
}

@article{mitro2015cross,
  title={Cross-sectional associations between exposure to persistent organic pollutants and leukocyte telomere length among US adults in NHANES, 2001--2002},
  author={Mitro, Susanna D and Birnbaum, Linda S and Needham, Belinda L and Zota, Ami R},
  journal={Environmental health perspectives},
  volume={124},
  number={5},
  pages={651},
  year={2015}
}

@article{gibson2019overview,
  title={An overview of methods to address distinct research questions on environmental mixtures: an application to persistent organic pollutants and leukocyte telomere length},
  author={Gibson, Elizabeth A and Nunez, Yanelli and Abuawad, Ahlam and Zota, Ami R and Renzetti, Stefano and Devick, Katrina L and Gennings, Chris and Goldsmith, Jeff and Coull, Brent A and Kioumourtzoglou, Marianthi-Anna},
  journal={Environmental Health},
  volume={18},
  number={1},
  pages={76},
  year={2019},
  publisher={Springer}
}

@article{
	ngiangi_bakwin2001semiparametric,
	author = {Ngianga-Bakwin Kandala and Ludwig Fahrmeir and Stefan Lang and Stephan Klasen},
	title = {Semiparametric Analysis of the Socio-Demographic and Spatial Determinants of Chronic Undernutrition in Two African Countries},
	year = {2001},
    journal={Research in Official Statistics},
    volume={4},
    number={1},
    pages={81-100}
}

@article{Simpson2024, doi = {10.21105/joss.06962}, url = {https://doi.org/10.21105/joss.06962}, year = {2024}, publisher = {The Open Journal}, volume = {9}, number = {104}, pages = {6962}, author = {Simpson, Gavin L.}, title = {gratia: An R package for exploring generalized additive models}, journal = {Journal of Open Source Software} }

@book{Paulsen_Raghupathi_2016, place={Cambridge}, series={Cambridge Studies in Advanced Mathematics}, title={An Introduction to the Theory of Reproducing Kernel Hilbert Spaces}, publisher={Cambridge University Press}, author={Paulsen, Vern I. and Raghupathi, Mrinal}, year={2016}, collection={Cambridge Studies in Advanced Mathematics}}

@article{duvenaud2011additivegaussianprocesses,
      title={Additive Gaussian Processes}, 
      author={David Duvenaud and Hannes Nickisch and Carl Edward Rasmussen},
      year={2011},
      eprint={1112.4394},
      archivePrefix={arXiv},
      primaryClass={stat.ML},
      url={https://arxiv.org/abs/1112.4394}, 
      journal={arXiv preprint arXiv:1112.4394}
}

@article{sobol1990sensitivity,
  title={On sensitivity estimation for nonlinear mathematical models},
  author={Sobol', Il'ya Meerovich},
  journal={Matematicheskoe modelirovanie},
  volume={2},
  number={1},
  pages={112--118},
  year={1990},
  publisher={Russian Academy of Sciences, Branch of Mathematical Sciences}
}

@article{eaves2022analysis,
  title={Analysis of the novel NCWELL database highlights two decades of co-occurrence of toxic metals in North Carolina private well water: Public health and environmental justice implications},
  author={Eaves, Lauren A and Keil, Alexander P and Rager, Julia E and George, Andrew and Fry, Rebecca C},
  journal={Science of the Total Environment},
  volume={812},
  pages={151479},
  year={2022},
  publisher={Elsevier}
}

@book{dieter2018estimated,
  title={Estimated use of water in the United States in 2015},
  author={Dieter, Cheryl A and Maupin, Molly A and Caldwell, Rodney R and Harris, Melissa A and Ivahnenko, Tamara I and Lovelace, John K and Barber, Nancy L and Linsey, Kristin S},
  number={1441},
  year={2018},
  publisher={US Geological Survey}
}

@misc{astdr_priority,
title={Substance Priority List},
author={ATSDR},
publisher={ATSDR},
year={2025}}

@article{naujokas2013broad,
  title={The broad scope of health effects from chronic arsenic exposure: update on a worldwide public health problem},
  author={Naujokas, Marisa F and Anderson, Beth and Ahsan, Habibul and Aposhian, H Vasken and Graziano, Joseph H and Thompson, Claudia and Suk, William A},
  journal={Environmental health perspectives},
  volume={121},
  number={3},
  pages={295},
  year={2013}
}

@article{eick2020arsenic,
  title={Arsenic and obesity: a review of causation and interaction},
  author={Eick, Stephanie M and Steinmaus, Craig},
  journal={Current environmental health reports},
  volume={7},
  number={3},
  pages={343--351},
  year={2020},
  publisher={Springer}
}

@article{plescia2014reducing,
  title={Reducing health disparities by addressing social determinants of health: the Mecklenburg County experience},
  author={Plescia, Marcus and Emmanuel, Cheryl},
  journal={North Carolina Medical Journal},
  volume={75},
  number={6},
  pages={417--421},
  year={2014},
  publisher={NCIOM, The Duke Endowment}
}

@article{barber2020place,
  title={Place matters: from health and health care disparities to equity and liberation},
  author={Barber, Sharrelle},
  journal={North Carolina Medical Journal},
  volume={81},
  number={3},
  pages={173--176},
  year={2020},
  publisher={NCIOM, The Duke Endowment}
}

@article{greenlund2022places,
  title={PLACES: local data for better health},
  author={Greenlund, Kurt J and Lu, Hua and Wang, Yan and Matthews, Kevin A and LeClercq, Jennifer M and Lee, Benjamin and Carlson, Susan A},
  journal={Preventing chronic disease},
  volume={19},
  pages={E31},
  year={2022}
}

@misc{acs2019_5yr,
  author       = {{U.S. Census Bureau}},
  title        = {American Community Survey 2015--2019 5-Year Estimates},
  year         = {2020},
  publisher    = {U.S. Department of Commerce},
  howpublished = {\url{https://data.census.gov/}},
  note         = {Accessed: 2026-07-23}
}

@article{zangiabadian2024association,
  title={The association between heavy metal exposure and obesity: A systematic review and meta-analysis},
  author={Zangiabadian, Moein and Jolfayi, Amir Ghaffari and Nejadghaderi, Seyed Aria and Amirkhosravi, Ladan and Sanjari, Mojgan},
  journal={Journal of Diabetes \& Metabolic Disorders},
  volume={23},
  number={1},
  pages={11--26},
  year={2024},
  publisher={Springer}
}

@article{tinkov2021adipotropic,
  title={Adipotropic effects of heavy metals and their potential role in obesity},
  author={Tinkov, Alexey A and Aschner, Michael and Ke, Tao and Ferrer, Beatriz and Zhou, Ji-Chang and Chang, Jung-Su and Santamar{\'\i}a, Abel and Chao, Jane CJ and Aaseth, Jan and Skalny, Anatoly V},
  journal={Faculty reviews},
  volume={10},
  pages={32},
  year={2021}
}

@article{lakka2023potential,
  title={Potential diagnostic biomarkers for lead-induced hepatotoxicity and the role of synthetic chelators and bioactive compounds},
  author={Lakka, Netranandini and Pai, Bhagyashree and Mani, Monica Shirley and Dsouza, Herman Sunil},
  journal={Toxicology Research},
  volume={12},
  number={2},
  pages={178--188},
  year={2023},
  publisher={Oxford University Press}
}

@book{casella2024statistical,
  title={Statistical inference},
  author={Casella, George and Berger, Roger},
  year={2024},
  publisher={CRC press}
}

@book{boos2013essential,
  title={Essential statistical inference: theory and methods},
  author={Boos, Dennis D and Stefanski, Leonard A},
  volume={591},
  year={2013},
  publisher={Springer}
}

@book{agresti2015foundations,
  title={Foundations of linear and generalized linear models},
  author={Agresti, Alan},
  year={2015},
  publisher={John Wiley \& Sons}
}

@article{kong2020use,
  title={The use of small area estimates in place-based health research},
  author={Kong, Amanda Y and Zhang, Xingyou},
  journal={American Journal of Public Health},
  volume={110},
  number={6},
  pages={829--832},
  year={2020},
  publisher={American Public Health Association}
}

@article{zhang2014multilevel,
  title={Multilevel regression and poststratification for small-area estimation of population health outcomes: a case study of chronic obstructive pulmonary disease prevalence using the behavioral risk factor surveillance system},
  author={Zhang, Xingyou and Holt, James B and Lu, Hua and Wheaton, Anne G and Ford, Earl S and Greenlund, Kurt J and Croft, Janet B},
  journal={American journal of epidemiology},
  volume={179},
  number={8},
  pages={1025--1033},
  year={2014},
  publisher={Oxford University Press}
}

@article{murray2025machine,
  title={A machine learning approach to estimate domestic use of public and private water sources in the United States},
  author={Murray, Andrew and Hall, Alexander and Riveros-Iregui, Diego},
  journal={Water research},
  volume={276},
  pages={123171},
  year={2025},
  publisher={Elsevier}
}

@article{hooker2007generalized,
  title={Generalized functional anova diagnostics for high-dimensional functions of dependent variables},
  author={Hooker, Giles},
  journal={Journal of computational and graphical statistics},
  volume={16},
  number={3},
  pages={709--732},
  year={2007},
  publisher={Taylor \& Francis}
}

@article{GpGp,
author = {Joseph Guinness},
title = {{Permutation and grouping methods for sharpening Gaussian Process approximations}},
journal = {Technometrics},
volume = {60},
number = {4},
pages = {415--429},
year = {2018},
publisher = {Taylor \& Francis},
doi = {10.1080/00401706.2018.1437476},

    note ={PMID: 31447491},


URL = { 
    
        https://doi.org/10.1080/00401706.2018.1437476
    
    

},
eprint = { 
    
        https://doi.org/10.1080/00401706.2018.1437476
    
    

}

}
\clearpage 

\begin{center}
    \LARGE\textbf{Supplemental Materials for: Hypothesis Testing for Nonlinear and Interactive Effects Using Additive Gaussian Processes}
\end{center}
\vspace{2em}

\setcounter{equation}{0}
\setcounter{figure}{0}
\setcounter{table}{0}
\setcounter{section}{0}
\setcounter{page}{1} 

\renewcommand{\theequation}{S\arabic{equation}}
\renewcommand{\thefigure}{S\arabic{figure}}
\renewcommand{\thetable}{S\arabic{table}}
\renewcommand{\thesection}{S\arabic{section}}

\renewcommand{\theHequation}{S\arabic{equation}}
\renewcommand{\theHfigure}{S\arabic{figure}}
\renewcommand{\theHtable}{S\arabic{table}}
\renewcommand{\theHsection}{S\arabic{section}}

\section{GKM as a test of linear association}\label{s:supp_gkm_lin}

\subsection{Summary}
The garrote kernel used in \cite{maity2011powerful} to test $H_0$: $x_1$ unassociated with $y$ is $$K_g(\mathbf{x}, \mathbf{x}')=\exp\left\{\frac{-1}{\rho}\left(\delta(x_1 - x_1')^2 + \sum_{j=2}^p (x_j - x_j')^2 \right)\right\}.$$ Using this garrote kernel, at $\delta=0$, $x_1$ is removed from the model, as in our approach. However, using this garrote parameter has the result of simultaneously determining $x_1$'s removal or inclusion in the model, and smoothing the model with respect to $x_1$. At a high level, the score test for $H_0$: $\delta=0$ is comparing the fit of two models: one with $x_1$ removed, and one with $\delta$ infinitesimally larger than 0, which includes $x_1$ but is essentially infinitely smoothed. This infinitely smoothed model turns out to allow only a linear association between $x_1$ and $y$, as we explain in more detail below.

\subsection{Detailed explanation}
The test statistic used in \cite{maity2011powerful}, \cite{liu2022cross}, and our paper is

$$T_n = (\mathbf{Y}-\mathbf{Z}\tilde{\boldsymbol{\beta}})^T\mathbf{V}_0^{-1}\left(\frac{\partial}{\partial \delta}\mathbf{V}_{\tilde{\theta}}  \right)\mathbf{V}_0^{-1}(\mathbf{Y}-\mathbf{Z}\tilde{\boldsymbol{\beta}}),$$

which in \cite{maity2011powerful} is

$$T_n = (\mathbf{Y}-\mathbf{Z}\tilde{\boldsymbol{\beta}})^T\mathbf{V}_0^{-1}\left(-\lambda\mathbf{D}\odot\boldsymbol{\Sigma}_{-1} \right)\mathbf{V}_0^{-1}(\mathbf{Y}-\mathbf{Z}\tilde{\boldsymbol{\beta}}),$$
where $\mathbf{D}_{ij}=(x_{i1}-x_{j1})^2$, $\odot$ is the element-wise (Hadamard) product, and $\boldsymbol{\Sigma}_{-1}$ is the squared-exponential covariance computed without $\mathbf{x}_1$. We can write $-\mathbf{D}= 2\mathbf{x}_1\mathbf{x}_1^T - (\mathbf{x}_1\odot\mathbf{x}_1)\mathbf{1}^T - \mathbf{1}(\mathbf{x}_1\odot\mathbf{x}_1)^T$. Then $-\lambda\mathbf{D}\odot\boldsymbol{\Sigma}_{-1} = \lambda\left( 2\mathbf{x}_1\mathbf{x}_1^T - (\mathbf{x}_1\odot\mathbf{x}_1)\mathbf{1}^T - \mathbf{1}(\mathbf{x}_1\odot\mathbf{x}_1)^T \right)\odot \boldsymbol{\Sigma}_{-1} =2\lambda\mathbf{x}_1\mathbf{x}_1^T \odot \boldsymbol{\Sigma}_{-1} - \lambda\text{diag}(\mathbf{x}_1^2)\boldsymbol{\Sigma}_{-1} - \lambda \boldsymbol{\Sigma}_{-1} \text{diag}(\mathbf{x}_1^2).$ This is identical to the test statistic if $\frac{\partial}{\partial \delta}\mathbf{V}_{\tilde{\theta}}=2\lambda\mathbf{x}_1\mathbf{x}_1^T \odot \boldsymbol{\Sigma}_{-1} - \lambda\text{diag}(\mathbf{x}_1^2)\boldsymbol{\Sigma}_{-1} - \lambda \boldsymbol{\Sigma}_{-1} \text{diag}(\mathbf{x}_1^2)$. Notably, this is the case if the garrote $\delta$ multiplies the matrix $2\lambda\mathbf{x}_1\mathbf{x}_1^T \odot \boldsymbol{\Sigma}_{-1} - \lambda\text{diag}(\mathbf{x}_1^2)\boldsymbol{\Sigma}_{-1} - \lambda \boldsymbol{\Sigma}_{-1} \text{diag}(\mathbf{x}_1^2)$, which we will denote $\boldsymbol{\Sigma}'_1$. We can understand the test in \cite{maity2011powerful} by understanding the terms of $\boldsymbol{\Sigma}'_1$, because if $\delta$ multiplies $\boldsymbol{\Sigma}_1'$, these terms are being tested for importance to the model, as in our additive GP approach. The first is a Hadamard product between a linear kernel for $\mathbf{x}_1$ and a squared-exponential kernel for the remaining terms; this models a tensor-product interaction between linear functions of $x_1$ and nonlinear functions of the remaining predictors. The second two terms are simply rescalings of $\boldsymbol{\Sigma}_{-1}$ and therefore do not model effects of $x_1$. As a result, the first term is the only one modeling effects of $x_1$ and it restricts them to be linear.


\section{Methods: Additional implementation details and discussion}\label{s:supp_pagp_details}

\subsection{Orthogonalization of main and interaction effects}\label{s:orth}

We use the standard scheme described in \cite{gu2013smoothing} and \cite{wood2017generalized} to orthogonalize the main and interaction effects using computed matrices rather than the kernel functions. 
%
After computing the 1-D base covariance matrices $\boldsymbol{\Sigma}_j$, these are modified so that all estimates of $f_j$ sum to zero, by column- and row-centering; denote this centered version by $\tilde{\boldsymbol{\Sigma}}_j$. Then, $\tilde{\boldsymbol{\Sigma}}_{jk}$ is the Hadamard (element-wise) product of $\tilde{\boldsymbol{\Sigma}}_j$ and $\tilde{\boldsymbol{\Sigma}}_k$, $\boldsymbol{\Sigma}_{jk}=\tilde{\boldsymbol{\Sigma}}_j \odot \tilde{\boldsymbol{\Sigma}}_k$. This generates a tensor product space of functions that is orthogonal to the space of main effects generated by the 1-D kernels \citep{gu2013smoothing, wood2017generalized}. This is similar to the orthogonalized kernel in \cite{lu2022oak} based on the empirical distribution of $\mathbf{x}$. This implicitly assumes that the predictors are uncorrelated \citep{lu2022oak}.

\subsection{Approximate null distribution}\label{s:tn_null}

Following \cite{zhang2003hypothesis} and \cite{maity2011powerful}, we use the Satterthwaite approximation to approximate the null distribution of $T_n=(\mathbf{Y}-\mathbf{Z}\tilde{\boldsymbol{\beta}})^T\mathbf{V}_0^{-1}\biggl(\boldsymbol{\Sigma}_1  \biggr)\mathbf{V}_0^{-1}(\mathbf{Y}-\mathbf{Z}\tilde{\boldsymbol{\beta}})$ using and $m\chi^2_d$ random variable, where parameters accented with a tilde are estimated under the null hypothesis, i.e., $\delta$ is restricted to be equal to zero. Summarizing \cite{maity2011powerful}, the information matrix based on the REML criterion of $\boldsymbol{\theta}$ is $\mathbf{I} = \begin{bmatrix}I_{\delta\delta} & I_{\delta\tilde{\theta}^T}\\ I_{\delta \tilde{\theta}} & I_{\tilde{\theta}\tilde{\theta}} \end{bmatrix}$, where $\mathbf{I}_{\theta_j\theta_k}=\text{tr}\{\mathbf{P}_{\tilde{\boldsymbol{\theta}}}\frac{\partial \mathbf{V}_{\boldsymbol{\theta}}}{\partial \theta_j}\mathbf{P}_{\tilde{\boldsymbol{\theta}}}\frac{\partial \mathbf{V}_{\boldsymbol{\theta}}}{\partial \theta_k}\}/2$, and $\mathbf{P}_{\tilde{\boldsymbol{\theta}}} = \boldsymbol{V}^{-1}_{\tilde{\boldsymbol{\theta}}} - \boldsymbol{V}^{-1}_{\tilde{\boldsymbol{\theta}}}\mathbf{Z}^T\{\mathbf{Z}\boldsymbol{V}^{-1}_{\tilde{\boldsymbol{\theta}}}\mathbf{Z}^T\}^{-1}\mathbf{Z}\boldsymbol{V}^{-1}_{\tilde{\boldsymbol{\theta}}}$. Then the efficient information for $\delta$, accounting for estimation of other parameters, is $\widehat{\mathbf{I}}_{\delta \delta} = \mathbf{I}_{\delta\delta}-\mathbf{I}_{\delta\tilde{\theta}}^T\mathbf{I}_{\tilde{\theta}\tilde{\theta}}^{-1}\mathbf{I}_{\delta\tilde{\theta}}$. To match moments of $m\chi^2_d$ and $T_n$, solve $md=\text{tr}\{\mathbf{P}_{\tilde{\boldsymbol{\theta}}}\boldsymbol{\Sigma}_1^*\}/2$ and $2m^2d=\widehat{\mathbf{I}}_{\delta \delta}$, where $\boldsymbol{\Sigma}_1^*$ is defined in Section 3.2.1. Note that because in practice we estimate $\rho$ under $H_1$ and then treat it as fixed (\ref{s:rho}), $\mathbf{I}$ in practice omits entries related to $\rho$.

\subsection{Tests for interaction and joint association}\label{s:extensions}

To test for interaction between $x_j$ and $x_k$, the model is rearranged so that $\delta$ only multiplies the single covariance matrix responsible for modeling the interaction between $x_j$ and $x_k$:
\begin{align*}
    \mathbf{Y} &\sim N(\mathbf{0}, \delta \mathbf{\tilde{\Sigma}}_{jk} + \lambda \mathbf{\tilde{\Sigma}}_{\text{rest}} + \sigma^2\mathbf{I}),
\end{align*}
where $\mathbf{\tilde{\Sigma}}_{jk}$ is defined in \ref{s:orth} and $\mathbf{\tilde{\Sigma}}_{\text{rest}} $ is the sum of all other main effect and interaction covariance base matrices. Then the test of $H_0$: $\delta=0$ is carried out in the same way as for variable importance. This is essentially the test proposed in \cite{liu2022cross} but for pairwise interaction, rather than interaction between groups of variables, and using our pairwise-interactive model for other predictor effects. A global test for any interaction effects can be performed by specifying the model so that $\delta$ multiplies all of the interaction kernels:
\begin{align*}
    \mathbf{Y} &\sim N(\mathbf{0}, \delta \mathbf{\tilde{\Sigma}}_{\text{int}} + \lambda \mathbf{\tilde{\Sigma}}_{\text{main}} + \sigma^2\mathbf{I}),
\end{align*}
where $\mathbf{\tilde{\Sigma}}_{\text{int}}=\sum_{j<k}\mathbf{\tilde{\Sigma}}_{jk}$ is the sum of interaction effect base covariances, and $\mathbf{\tilde{\Sigma}}_{\text{main}}=\sum_{j=1}^p\mathbf{\tilde{\Sigma}}_{j}$ is the sum of main effect base covariances, as defined in \ref{s:orth}. Finally, to test for joint association between $\{x_j: j \in J\}$ for a set of indices $J$, and $y$, we use the model
\begin{align*}
    \mathbf{Y} &\sim N(\mathbf{0}, \delta\left(\sum_{j\in J} \mathbf{\tilde{\Sigma}}_{\text{j}} + \sum_{j,k\in J; j<k}\mathbf{\tilde{\Sigma}}_{\text{jk}} + \sum_{j\in J, k\notin J}\mathbf{\tilde{\Sigma}}_{\text{jk}}\right) + \lambda \mathbf{\tilde{\Sigma}}_{\text{rest}} + \sigma^2\mathbf{I}),
\end{align*}
where $\delta$ now multiplies all main effect and interaction covariance base matrices involving any of the $J$ predictors being tested for significance, and $\mathbf{\tilde{\Sigma}}_{\text{rest}}$ contains all other terms.

\subsection{Avoiding re-estimation of $\rho$}\label{s:rho}
Our data analysis procedure involves testing many hypotheses--of variable importance for $x_1$ through $x_p$, and then for pairwise interactions between any identified important variables, and possibly of multiple variable or interaction effects. Following the procedure in Section 3.2.1, it is necessary to re-estimate the lengthscale $\rho$ for each null hypothesis, typically the most computationally expensive part of estimation. There is little reason to expect that lengthscales estimated under the null are optimal, or even desirable, for hypothesis testing, as the lengthscale used to test the importance of $x_1$ is estimated exclusively using the other predictors, risking under- or over-smoothing with respect to $x_1$. To save computation, and reduce the risk of over- or undersmoothing when conducting hypothesis tests, we treat $\rho$ as a tuning parameter, estimate it once using REML and including all predictors using model (5) and with $\delta=\lambda$, and then treat this as fixed for the hypothesis tests, re-estimating of $\boldsymbol{\beta}$, $\lambda$, and $\sigma^2$ under the null hypothesis for each test.

\subsection{Multiple testing}\label{s:multi_test}

In a data analysis, we propose testing importance of all $p$ predictors, and then testing for pairwise interaction effects between the predictors identified as important. Joint tests of variable importance can be used if there is a natural grouping of the predictors, to reduce the number of hypothesis tests conducted in a first stage and reduce the family-wise Type I error rate; if joint tests identify groups of predictors as important, individual tests of importance can be used next. If there are some predictors of interest and covariates, a natural first step is to conduct a joint test of association between the predictors of interest and the outcome. Similarly, if there are many important predictors, before conducting many tests of interaction between each pair of significant predictors, a global interaction test can be carried out prior to conducting tests of individual interaction effects.

A conservative way to adjust for multiple testing is the Bonferroni correction. Alternatively, \cite{maity2011powerful} proposed a permutation-based method for adjusting p-values, but this is computationally expensive, especially in our data analyses with over 1,000 observations and many predictors. For data analyses where permutation-based methods are infeasible, we recommend using joint tests if natural groupings of predictors exist, the test of global interaction if many tests of pairwise interactions are to be conducted, and Bonferroni correction otherwise.

\subsection{Computing}\label{s:computing}

Although the pairwise-interactive kernel is a sum with $p+{p \choose 2}=(p+p^2)/2$ terms, it can be computed in $O(p)$ operations using a special case of the approach used in \cite{duvenaud2011additivegaussianprocesses}. Define each univariate kernel by $K_j(\mathbf{x}, \mathbf{x}')=\exp\left( -(x_j-x_j')^2/\rho \right)$. Then $K_{j,k}=K_j\times K_k$. Then $$\sum_{j<k}K_{j,k}=\sum_{j<k}K_j\times K_k = \frac{1}{2}\left( \left(\sum_{j=1}^p K_j\right)^2-\sum_{j=1}^pK_j^2 \right).$$ These two summations can be computed in $O(p)$ time. The computation can be carried out using the base covariance matrices $\tilde{\boldsymbol{\Sigma}}_j$. Only the $p$ base 1-D covariance matrices need to be orthogonalized as in Section \ref{s:orth}, so orthogonalization also requires $O(p)$ operations. Each matrix operation is $O(n^2)$, so the total number of operations required to evaluate the kernel is $O(n^2p)$. Gradients with respect to covariance parameters can be efficiently computed similarly. Inverting the covariance matrix during optimization still requires $O(n^3)$ operations.

\subsection{Simplifying assumptions}\label{s:assns}
In the model in Section 3.2.1, we use a common lengthscale $\rho$ for all predictors $x_j$, a common variance parameter $\lambda$ for all effects not including $x_1$, and a common variance parameter $\delta$ for all effects including $x_1$. These strong simplifying assumptions greatly reduce the number of parameters that must be estimated compared to additive GPs whose goal is predictive accuracy such as \cite{duvenaud2011additivegaussianprocesses, durrande2012additive, lu2022oak}. It is also not possible to estimate a more richly parameterized model under $H_0$: $\delta=0$ because some parameters may not be identifiable under the null, for example, a lengthscale parameter specific to $x_1$ cannot be estimated when $x_1$ is removed from the model. 
While these simplifying assumptions reduce predictive accuracy of our model, our simulations show good hypothesis testing performance.
\section{Proofs}\label{s:supp_proof}

We assume the additive GP on $\mathbb{R}^p$ is of the form $W=\sum_{l=1}^L W_l$, where the $W_l$ are GPs with Mat\`ern covariance $K_l$. Denote by $\mathbf{x}_l$ a subvector of $\mathbf{x}$ specific to each $l$ with length $p_l$, and let $K_l:\mathbb{R}^p \times \mathbb{R}^p \rightarrow \mathbb{R}$ be defined as $\overline{K}_l(\mathbf{x}_l, \mathbf{x}_l')$ for some covariance function $\overline{K}:\mathbb{R}^{p_l} \times \mathbb{R}^{p_l} \rightarrow \mathbb{R}$. For example, for an additive GP including terms for each pair of predictors, and Mat\`ern covariance with some smoothness, the $\overline{K}_l$ would be equal to the 2-D Mat\`ern covariance function applied to each combination of two indices $j,k$ in $1, ..., p$.

\subsection{Notation and definitions}

We use the same notation as in \cite{vandervaart11a}. The empirical norm is defined as $\|f\|_n=(\frac{1}{n}\sum_{i=1}^nf^2(x_i))^{1/2}.$ If $a \leq Cb$ for some constant $C$, we write $a \lesssim b$. $\max\{a,b\}$ is denoted $a \vee b$ and $\min\{a,b\}$ is denoted $a \wedge b$. The reproducing kernel Hilbert space (RKHS) associated with a kernel $K$ is denoted $\mathcal{H}_K$, and the norm associated with this RKHS is denoted $\| \cdot \|_{\mathcal{H}_K}$.

We use $C^\alpha(\mathcal{X})$ to denote the H\"older space of functions, and $H^\alpha(\mathcal{X})$ to denote the Sobolev space of functions, $f_0: \mathcal{X}\rightarrow \mathbb{R}$. See \cite{vandervaart11a} for definitions of these spaces; loosely speaking, they are each spaces of functions that have $\alpha$ partial derivatives, plus additional constraints on the regularity of those derivatives.

\subsection{Contraction rate of additive GP}\label{s:contract_proof}

We use the results of \cite{van2008rates} and \cite{vandervaart11a}, combined with results on composition kernels in \cite{Paulsen_Raghupathi_2016}. 

By Theorem 1 of \cite{vandervaart11a}, it is sufficient to bound the concentration function of a GP in order to obtain its contraction rate for fixed design. In \cite{vandervaart11a}, for a fixed target function $f_0$, GP $W$ with kernel $K$, GP prior probability function $\Pi_W$, and RKHS associated with the GP $\mathcal{H}_K$, the concentration function is defined as
\begin{align*}
    \phi_{f_0}^W(\epsilon) = \inf_{h\in\mathcal{H}_K: \|h - f_0\|_{\infty}<\epsilon}\|h\|^2_{\mathcal{H}_K} - \log \Pi_W (f : \|f\|_{\infty}<\epsilon).
\end{align*}
An associated function is defined as
\begin{align*}
    \psi_{f_0}^W(\epsilon) = \frac{\phi_{f_0}^W(\epsilon)}{\epsilon^2},
\end{align*}
and its generalized inverse is defined as $(\psi_{f_0}^W)^{-1}(z) = \sup\{\epsilon>0:\psi_{f_0}^W(\epsilon) \geq z\}$. The logarithmic term of the concentration function is the exponent of the small-ball probability of the GP $W$, and the infimum measures the decrease in prior mass when the small ball is shifted to the target function $f_0$. Theorem 1 of \cite{vandervaart11a} states that if $f_0$ is a continuous and bounded function on a compact metric space $\mathcal{X}$, and predictors are considered fixed, then
\begin{align}\label{e:contract}
    E_{f_0}\int \|f - f_0\|_n^2d\Pi_n(f|Y_{1:n}) \lesssim \left( (\psi_{f_0}^W)^{-1}(n) \right)^2,
\end{align}
which by Markov's inequality shows the posterior contracts around $f_0$ at the rate $(\psi_{f_0}^W)^{-1}(n)$ in the empirical norm.

To show this, let the contraction rate be $\epsilon_n = (\psi_{f_0}^W)^{-1}(n)$. Theorem 1 implies that $E_{f_0} \int \|f - f_0\|_n^2 \, d\Pi_n(f \mid Y_{1:n}) \leq C\epsilon_n^2$ for some constant $C > 0$. 

To satisfy the formal definition of posterior contraction, we must show that the posterior mass outside a sequence of neighborhoods of size $M_n\epsilon_n$ vanishes, where $M_n \to \infty$ as $n \to \infty$. By Markov's inequality:
\[
\Pi_n(\|f - f_0\|_n > M_n\epsilon_n \mid Y_{1:n}) \leq \frac{1}{M_n^2\epsilon_n^2} \int \|f - f_0\|_n^2 \, d\Pi_n(f \mid Y_{1:n})
\]

Taking the frequentist expectation $E_{f_0}$ on both sides and substituting the quadratic risk bound:
\[
E_{f_0} \Pi_n(\|f - f_0\|_n > M_n\epsilon_n \mid Y_{1:n}) \leq \frac{1}{M_n^2\epsilon_n^2} (C\epsilon_n^2) = \frac{C}{M_n^2}
\]

As $n \to \infty$, the sequence $M_n \to \infty$. Therefore, the upper bound $C/M_n^2 \to 0$, formally establishing posterior contraction at the rate $\epsilon_n$.

\cite{vandervaart11a} established contraction rates for the Mat\'ern kernel, which we will use.

We now consider the additive GP, $W=\sum_{l=1}^LW_l$, each with covariance function $K_l$ and each $W_l$ defined on the same space $\mathcal{X}$. Assume that the target function $f_0$ can be written as a sum of functions, $f_0 = \sum_{l=1}^Lf_l$, where each $f_l$ must have domain $\mathcal{X}$. Assume that for each $l=1, ..., L$, we can define a subvector $\mathbf{x}_l$ of $\mathbf{x}$ with length $p_l$, and a function $g_l:\mathbb{R}^{p_l}\rightarrow \mathbb{R}$, such that $f_l(\mathbf{x}) =g_l(\mathbf{x}_l)$, in other words, the value of $f_l$ is only determined by the elements of the subvector $\mathbf{x}_l$. Let each $W_l$ have concentration function $\phi^{W_l}_{f_l}$ for $f_{l}$ in the same separable Banach space of functions. Theorem 2.3 of \cite{van2008rates} states that the concentration function of $W$ satisfies
\begin{align*}
    \phi_{f_0}^W(\epsilon L) \leq 2\sum_{l=1}^L\phi^{W_l}_{f_l}(\epsilon/2),
\end{align*}
in other words, if we can bound the individual concentration functions of each $W_l$ with respect to the corresponding $f_{l}$, we can bound the concentration function of the additive GP $W$ and obtain its contraction rate. In Lemma \ref{lem:gp_component}, we bound the concentration function for the component GPs $W_l$.

\begin{lemma}\label{lem:gp_component}
    Let $W_l$ be a GP on $\mathbb{R}^p$. Let $\mathbf{x}_l$ be a subvector of length $p_l$ of $\mathbf{x} \in\mathcal{X}\subset \mathbb{R}^p$, with $\mathbf{x}_l\in\mathcal{X}_l\subset\mathbb{R}^{p_l}$. Let $f_{l}$ be a function on $\mathbb{R}^p$ that is both Sobolev and H\"older smooth with smoothness $\beta_l$, and whose values depend only on the values of $\mathbf{x}_l$. In other words, for some function ${g}_{l}:\mathbb{R}^{p_l}\rightarrow\mathbb{R}$ and function $\varphi_l(\mathbf{x})=\mathbf{x}_l$, $f_{l} = {g}_{l}\circ \varphi_l$. Let the covariance of $W_l$, $K_l$, be defined by the Mat\'ern covariance function with smoothness $\alpha_l$ applied to $\mathbf{x}_l$. Then the concentration function of $W_l$ satisfies
    \begin{align*}
        \phi_{f_{l}}^{W_l}(\epsilon) \lesssim \left(\frac{1}{\epsilon}\right)^{(2\alpha_l + p_l - 2\beta_l)/\beta_l} + \left(\frac{1}{\epsilon}\right)^{p_l/\alpha_l}.
    \end{align*}
\end{lemma}

\begin{proof}
    Define $W_l$ by letting its sample paths be $W_l(\mathbf{x})=\overline{W}_l(\mathbf{x}_l)$, where $\overline{W}_l$ is a GP on $\mathbb{R}^{p_l}$ with covariance function $\overline{K}_l$ and $\mathbf{x}_l$ is a subvector of $\mathbf{x}$ as defined above. Then $W_l$ is a GP, and the covariance of $W_l$, denoted $K_l$, can be written as a function from $\mathbb{R}^p \times \mathbb{R}^p \rightarrow \mathbb{R}$ by writing it as a composition, $K_l = \overline{K}_l \circ \varphi_l$.

    By Proposition 5.6 of \cite{Paulsen_Raghupathi_2016} $K_l$ is a kernel function, and by Theorem 5.7 of \cite{Paulsen_Raghupathi_2016}, $\mathcal{H}_{K_l} = \{ h_l \circ \varphi_l: h_l \in \mathcal{H}_{\overline{K}_l} \}$, and for $u \in \mathcal{H}_{K_l}, \|u\|_{\mathcal{H}_{K_l}} = \min \{\|f\|_{\mathcal{H}_{\overline{K}_l}}: u = f\circ \varphi_l\}.$ Define $\overline{K}_l: \mathbb{R}^{p_l}\times \mathbb{R}^{p_l} \rightarrow \mathbb{R}$ to be the Mat\'ern covariance in $p_l$ dimensions with smoothness $\alpha_l$.

    Following Lemmas 3 and 4 of \cite{vandervaart11a} (results for Mat\'ern priors), we bound the two terms of the concentration function separately, rewriting our GPs and kernels until we can apply the bounds of Lemmas 3 and 4 of \cite{vandervaart11a}. 
    
    The small-ball exponent can be bounded by rewriting terms. $\Pi_{W_l} (f : \|f\|_{\infty}<\epsilon) = P(\|W_l\|_{\infty} <\epsilon) = P(\sup_{\mathbf{x}\in\mathcal{X}}|W_l(\mathbf{x})|<\epsilon) = P(\sup_{\mathbf{x}_l \in \mathcal{X}_l}|\overline{W}_l(\mathbf{x}_l)|<\epsilon)=P(\|\overline{W}_l\|_{\infty}<\epsilon)=\Pi_{\overline{W}_l}(\overline{f}:\|\overline{f}\|_{\infty}<\epsilon)$. With the small-ball probability of $W_l$ written as a small-ball probability of a standard Mat\'ern GP on $\mathbb{R}^{p_l}$, we can use the bound of $C\left(\frac{1}{\epsilon}\right)^{p_l/\alpha_l}$ for the small-ball exponent for Mat\'ern prior from Lemma 3, \cite{vandervaart11a}.

    For the infimum, by Theorem 5.7 of \cite{Paulsen_Raghupathi_2016}, $\mathcal{H}_{K_l}=\{h_l \circ \varphi_l: h_l \in \mathcal{H}_{\overline{K}_l}\}$, and for any $u = h_l \circ \varphi_l,$
    $$\|u\|_{\mathcal{H}_{K_l}} = \min\{ \|h\|_{H_{\overline{K}_l}}:u=h\circ\varphi_l\} \leq \|h_l\|_{H_{\overline{K}_l}}.$$ Then note that the approximation error in ${H_{{K}_l}}$, $\|u - f_l\|_{\infty, \mathcal{X}} = \sup_{\mathbf{x}\in\mathcal{X}}|h_l(\varphi_l(\mathbf{x})) - g_l(\varphi_l(\mathbf{x}))| = \|h_l - g_l\|_{\infty, \mathcal{X}_l}$. Then we have
    \begin{align*}
        \inf_{u \in \mathcal{H}_{K_l}: \|u-f_l\|_{\infty, \mathcal{X}}<\epsilon}\|u\|^2_{\mathcal{H}_{K_l}} &= \inf_{h_l\in\mathcal{H}_{\overline{K}_l}:\|h_l - g_l\|_{\infty, \mathcal{X}_l}<\epsilon}\|h_l \circ \varphi_l\|^2_{\mathcal{H}_{K_l}}\\
        &\leq \inf_{h_l\in\mathcal{H}_{\overline{K}_l}:\|h_l - g_l\|_{\infty, \mathcal{X}_l}<\epsilon}\|h_l\|^2_{\mathcal{H}_{\overline{K}_l}}.
    \end{align*}
    Lemma 4 of \cite{vandervaart11a} bounds this final quantity by a constant times $\left(\frac{1}{\epsilon}\right)^{(2\alpha_l + p_l-2\beta_l)/\beta_l}$. Therefore the concentration function for $W_l$ around $f_l$ has the same rate as the concentration function for $\overline{W}_l$ around $g_l$, which is found in \cite{vandervaart11a}.
    
\end{proof}

Theorem 1 is proved by combining Theorem 2.3 from \cite{van2008rates} with Lemma \ref{lem:gp_component}.

\begin{proof}[Proof of Theorem 1]
    By Lemma \ref{lem:gp_component}, for each component $l \in \{1, \dots, L\}$, there exists a constant $c_l > 0$ such that the concentration function evaluated at any arbitrary positive scalar $x$ is bounded by:
    \[
    \phi_{f_l}^{W_l}(x) \leq c_l \left[ \left(\frac{1}{x}\right)^{(2\alpha_l+p_l-2\beta_l)/\beta_l} + \left(\frac{1}{x}\right)^{p_l/\alpha_l} \right].
    \]
    
    Theorem 2.3 of \cite{van2008rates} states that for any $\epsilon > 0$, the concentration function of the additive GP satisfies $\phi_{f_0}^W(\epsilon L) \leq 2\sum_{l=1}^L\phi^{W_l}_{f_{l}}(\epsilon/2)$. 
    
    Let $\epsilon' = \epsilon L$, which implies $\epsilon = \epsilon'/L$. Substituting $\epsilon'$ into the theorem gives a bound with respect to $\epsilon'$:
    \[
    \phi_{f_0}^W(\epsilon') \leq 2\sum_{l=1}^L\phi^{W_l}_{f_{l}}\left(\frac{\epsilon'}{2L}\right).
    \]
    
    Applying the bound from Lemma \ref{lem:gp_component} by evaluating it at $x = \frac{\epsilon'}{2L}$, we obtain:
    \[
    \phi_{f_0}^W(\epsilon') \leq 2 \sum_{l=1}^L c_l \left[ \left(\frac{2L}{\epsilon'}\right)^{(2\alpha_l+p_l-2\beta_l)/\beta_l} + \left(\frac{2L}{\epsilon'}\right)^{p_l/\alpha_l} \right].
    \]
    
    To bound this sum, let $\gamma_l = \max \left\{ \frac{2\alpha_l+p_l-2\beta_l}{\beta_l}, \frac{p_l}{\alpha_l} \right\}$ be the maximum exponent for component $l$, and let $\gamma_{max} = \max_{l} \gamma_l$ be the largest exponent across all components. Let $c_{max} = \max_l c_l$ be the largest constant. Asymptotically, we consider sufficiently small $\epsilon'$ such that $\epsilon' < 2L$, ensuring $\frac{2L}{\epsilon'} > 1$. Under this condition, $\left(\frac{2L}{\epsilon'}\right)^x \leq \left(\frac{2L}{\epsilon'}\right)^{\gamma_{max}}$ for any $x \leq \gamma_{max}$. 
    
    We can now bound each term in the summation by the worst-case rate and constant. Replacing each component's bound with this maximum yields:
    \begin{align*}
        \phi_{f_0}^W(\epsilon') &\leq 2 \sum_{l=1}^L c_{max} \left[ \left(\frac{2L}{\epsilon'}\right)^{\gamma_{max}} + \left(\frac{2L}{\epsilon'}\right)^{\gamma_{max}} \right] \\
        &= 2 \sum_{l=1}^L 2 c_{max} \left(\frac{2L}{\epsilon'}\right)^{\gamma_{max}} \\
        &= 4 L c_{max} (2L)^{\gamma_{max}} \left(\frac{1}{\epsilon'}\right)^{\gamma_{max}}.
    \end{align*}
    
    Let $C^* = 4 L c_{max} (2L)^{\gamma_{max}}$, which is a finite positive constant independent of $\epsilon'$. Thus, we have the simplified bound:
    \[
    \phi_{f_0}^W(\epsilon') \leq C^* \left(\frac{1}{\epsilon'}\right)^{\gamma_{max}}.
    \]
    
    Following \cite{vandervaart11a}, the posterior contraction rate $\epsilon_n$ is obtained by finding a sequence that solves the relation $\phi_{f_0}^W(\epsilon_n) \leq n\epsilon_n^2$. By letting our target sequence $\epsilon_n$ take the role of $\epsilon'$ in our established bound, we set:
    \[
    C^* \epsilon_n^{-\gamma_{max}} = n\epsilon_n^2 \implies \epsilon_n^{\gamma_{max} + 2} = C^* n^{-1} \implies \epsilon_n = (C^*)^{1/(\gamma_{max} + 2)} n^{-1/(\gamma_{max} + 2)}.
    \]
    
    Note that minimizing the rate exponent $1/(\gamma_{max}+2)$ is mathematically equivalent to maximizing the concentration exponent $\gamma$. From the definitions of our exponents, calculating $1/(\gamma+2)$ for the two terms in Lemma \ref{lem:gp_component} yields $\beta_l/(2\alpha_l+p_l)$ and $\alpha_l/(2\alpha_l+p_l)$, respectively. Therefore, defining the final rate exponent as the minimum across all components:
    \[
    \eta = \min_{l=1, \dots L} \left\{ \frac{\alpha_l \wedge \beta_l}{2\alpha_l + p_l} \right\},
    \]
    we have $1/(\gamma_{max} + 2) = \eta$. 
    
    Thus, the sequence bounding the relation is $\epsilon_n \asymp n^{-\eta}$. By definition, $\psi_{f_0}^W(\epsilon) = \phi_{f_0}^W(\epsilon) / \epsilon^2$, meaning our relation $\phi_{f_0}^W(\epsilon_n) \leq n\epsilon_n^2$ is equivalent to $\psi_{f_0}^W(\epsilon_n) \leq n$. This implies that the generalized inverse satisfies $(\psi_{f_0}^W)^{-1}(n) \lesssim n^{-\eta}$. As established by Theorem 1 of \cite{vandervaart11a}, the rate of contraction of the posterior distribution is exactly of this order in the case of fixed design. Therefore, the relation (\ref{e:contract}) holds, establishing that the rate of contraction of the posterior of $W$ is of order $n^{-\eta}$.
\end{proof}

\subsection{Proof of Theorem 2}\label{s:size_proof}

\begin{proof}[Proof of Theorem 2]
    Theorem 2 is a direct consequence of established results bounding the metric entropy of spaces comprised of additive functions. Consider a general space of functions defined on a compact domain $\mathcal{S} \subset \mathbb{R}^p$ belonging to the Hölder class $\mathcal{H}_\alpha^p$, where the smoothness is given by $\alpha = m + \nu$ for an integer $m$ and $\nu \in (0, 1]$. As established by \cite{lorentz1962metric}, the size of this space can be quantified by its metric entropy, which is proportional to $\epsilon^{-p/\alpha}$. 
    
    This general rate characterizes the size of our component spaces. Specifically, if a function $f$ belongs to a space $\Theta_j$ such that $f \in C^\beta(\mathcal{X}_j) \cap H^\beta(\mathcal{X}_j)$, the metric entropy of $\Theta_j$ is proportional to $\epsilon^{-p_j/\beta}$.

    We now extend this to the additive structure used in the pairwise Gaussian process. Suppose $\Theta_j$ is a space of functions $f$ such that $f = \sum_{l=1}^L g_l$, where each component function $g_l$ operates on a 2-dimensional subspace and satisfies $g_l \in C^\beta(\mathcal{X}_l) \cap H^\beta(\mathcal{X}_l)$. By the component bound, each 2-D Hölder space has a metric entropy of order $\epsilon^{-2/\beta}$. 
    
    Following Equation (6) on page 484 of \cite{lorentz1962metric}, the metric entropy of a function space created by taking the direct sum of $L$ such 2-D spaces scales at the same rate as the individual component spaces. Because the number of additive components $L$ is finite, the metric entropy of the direct sum space $\Theta_j$ is bounded by $O(L\epsilon^{-2/\beta}) = O(\epsilon^{-2/\beta})$.
\end{proof}

\section{NHANES 2001-2002: Telomere data}\label{s:telomere}

This data set has been used as a case study for methods for environmental exposure mixtures \citep{mcgee2023bayesian, gibson2019overview}. The outcome variable is natural-log mean leukocyte telomere length (LTL), which indicates cellular aging \citep{mitro2015cross}. We use the same sample of 1,003 individuals described in \citep{mitro2015cross} used in the subsequent analyses \citep{gibson2019overview,mcgee2023bayesian}.  The exposure variables are 18 persistent organic pollutants (POPs) and per \cite{mcgee2023bayesian}, have a natural grouping based on hypothesized health effects. PCBs 74, 99, 138, 153, 170, 180, 187, and 194 are non-dioxin-like polychlorinated biphenyls; their concentrations in serum are provided in ng/g in NHANES data. PCBs 126 and 169 are non-ortho PCBs; their serum concentrations are in fg/g. PCB 118, and seven furans and dioxins (1,2,3,6,7,8-hxcdd; 1,2,3,4,6,7,8-hpcdd; 1,2,3,4,6,7,8,9-ocdd; 2,3,4,7,8-pncdf; 1,2,3,4,7,8-hxcdf; 1,2,3,6,7,8-hxcdf; 1,2,3,4,6,7,8-hxcdf) form the final group; PCB118 is in ng/g, while the remainder are in fg/g. All POPs were logged as in \cite{mitro2015cross}. We use the same covariates as \cite{mitro2015cross}: age (years) and age$^2$, sex, race/ethnicity, educational attainment, BMI category, smoking, and blood cell counts and percentages. See \cite{mitro2015cross} for further detail on the data and covariates. To match the previous analyses, we linearly adjusted for covariates and included the 18 POPs in the nonlinear components of the GAM, BKMR, GKM, and PAGP models. 100,000 MCMC iterations were used for BKMR and 50,000 discarded as burn-in.

Results for tests of group significance using OLS, BKMR, and our PAGP tests are reported in Table \ref{tab:nhanes_2001_joint}. GAM joint tests were not used because they are not strongly recommended in \cite{wood2017generalized} and are not widely-used, to our knowledge. The PAGP identifies groups 2 and 3 as significant, while no other methods show clear evidence of joint association between any groups and ln(LTL).
\begin{table}[!htbp]
\caption{P-values for tests of joint association between the POP exposure groups and log mean telomere length, for 1,003 individuals in NHANES 2001-2002 data, adjusted for age (years) and age$^2$, sex, race/ethnicity, educational attainment, BMI category, smoking, and blood cell counts and percentages. Because our focus is on hypothesis testing in nonparametric models, we report p-values rather than effect estimates and confidence intervals. Exposures were natural-log transformed.}
    \centering
    \begin{tabular}{clcc}
          & PAGP (p-value)&  OLS (p-value)& BKMR (PIP)\\
          Group 1&0.600&  0.712& 0.342\\
          Group 2 &0.012&  0.141& 0.646\\
 Group 3& 0.044& 0.206&0.759\\
    \end{tabular}
    
    \label{tab:nhanes_2001_joint}
\end{table}

Table \ref{tab:nhanes_telo_full} shows p-values and PIPs for three POPs identified as associated with ln(LTL) by at least one method. Because two predictors (PCB 126 and 2,3,4,7,8-pncdf, also called Furan1 in \cite{mcgee2023bayesian}) were associated with ln(LTL) based on the PAGP tests, we tested for pairwise interaction using the PAGP tests between PCB 126 and Furan1 and we found no evidence of interaction (p-value 0.68). To interpret the results, a GAM including covariates and additive nonlinear functions of PCB126 and Furan1 was fit, with results shown in Figure \ref{fig:telo}. PAGP and GKM yield similar results, with the different that PAGP can test for interaction. OLS and GAM found that PCB180 instead of PCB126 was associated with ln(LTL). BKMR found the PIP of Furan1 to be 0.79, but it is less clear how to interpret this PIP compared to p-values of other methods.

\begin{table}[!htbp]
\centering
\caption{Variable importance results (p-values for frequentist methods, PIPs for BKMR) for all exposures (log POP concentrations) in the NHANES 2001-2002 dataset, for 1,003 individuals in NHANES 2001-2002 data, adjusted for age (years) and age$^2$, sex, race/ethnicity, educational attainment, BMI category, smoking, and blood cell counts and percentages. Because our focus is on hypothesis testing in nonparametric models, we report p-values rather than effect estimates and confidence intervals. Exposures were natural-log transformed.}
\label{tab:nhanes_telo_full}
\begin{tabular}{lrrrrr}
\toprule
Exposure & \makecell[r]{OLS \\ (p-value)} & \makecell[r]{GAM \\ (p-value)} & \makecell[r]{GKM \\ (p-value)} & \makecell[r]{PAGP \\ (p-value)} & \makecell[r]{BKMR \\ (PIP)} \\
\midrule
log(PCB 74) & 0.795 & 0.697 & 0.741 & 0.858 & 0.125 \\
log(PCB 99) & 0.894 & 0.611 & 0.336 & 0.337 & 0.154 \\
log(PCB 138) & 0.556 & 0.578 & 0.692 & 0.543 & 0.151 \\
log(PCB 153) & 0.269 & 0.349 & 0.512 & 0.507 & 0.148 \\
log(PCB 170) & 0.654 & 0.507 & 0.843 & 0.437 & 0.161 \\
log(PCB 180) & 0.059 & 0.037 & 0.199 & 0.257 & 0.157 \\
log(PCB 187) & 0.995 & 1.000 & 0.901 & 0.729 & 0.125 \\
log(PCB 194) & 0.858 & 0.884 & 0.554 & 0.543 & 0.128 \\
log(PCB 126) & 0.229 & 0.201 & 0.030 & 0.029 & 0.311 \\
log(PCB 169) & 0.209 & 0.154 & 0.124 & 0.140 & 0.198 \\
log(PCB 118) & 0.478 & 0.441 & 0.107 & 0.234 & 0.243 \\
log(1,2,3,6,7,8-hxcdd) & 0.362 & 0.292 & 0.562 & 0.482 & 0.132 \\
log(1,2,3,4,6,7,8-hpcdd) & 0.682 & 0.605 & 0.955 & 0.748 & 0.107 \\
log(1,2,3,4,6,7,8,9-ocdd) & 0.183 & 0.315 & 0.226 & 0.269 & 0.114 \\
log(2,3,4,7,8-pncdf) & 0.022 & 0.028 & 0.004 & 0.002 & 0.790 \\
log(1,2,3,4,7,8-hxcdf) & 0.574 & 0.935 & 0.553 & 0.604 & 0.105 \\
log(1,2,3,6,7,8-hxcdf) & 0.322 & 0.085 & 0.863 & 0.381 & 0.140 \\
log(1,2,3,4,6,7,8-hxcdf) & 0.127 & 0.120 & 0.218 & 0.281 & 0.121 \\
\bottomrule
\end{tabular}
\end{table}
    
    

\begin{figure}[H]
    \centering
    
    \includegraphics[width=1\linewidth]{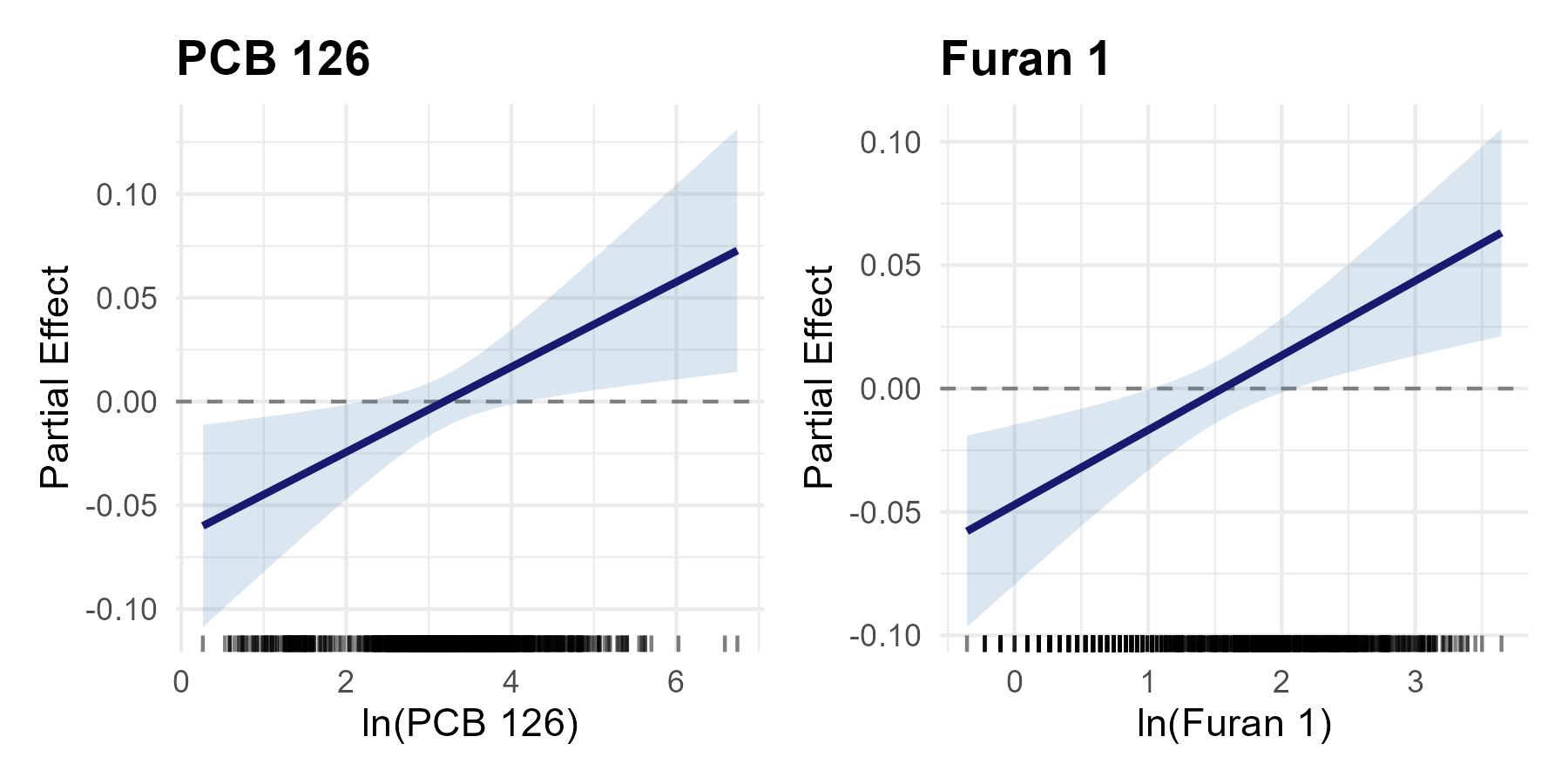}
    \caption{Main effects estimated by GAM smooth terms, using the significant main effects of PCB126 and Furan1 (2,3,4,7,8-pncdf)identified by the PAGP hypothesis tests, for the NHANES 2001-2002 data. The partial effects characterize associations found between natural log-transformed PCB in blood and log(LTL) in blood.}
    \label{fig:telo}
\end{figure}

\section{Additional data analysis results}\label{s:supp_addl_results}
\subsection{NCWELL Data}\label{s:ncwell_addl}

\begin{table}[!htbp]
\centering
\caption{NCWELL Continuous Variable Summary}
\label{tab:ncwell_cont_summary}
\begin{tabular}{lrrr}
\toprule
Characteristic & Mean & Median & SD \\
\midrule
Arsenic (mean, ppb) & 3.19 & 2.91 & 2.01 \\
Barium (mean, ppb) & 84.39 & 70.71 & 40.53 \\
Calcium (mean, ppb) & 24,660.78 & 18,005.06 & 20,278.43 \\
Chromium (mean, ppb)& 7.35 & 7.07 & 1.64 \\
Copper (mean, ppb) & 65.66 & 40.30 & 123.63 \\
Iron (mean, ppb)& 1,945.59 & 452.25 & 51,157.29 \\
Lead (mean, ppb)& 6.37 & 4.08 & 23.67 \\
Magnesium (mean, ppb)& 5,447.32 & 3,970.71 & 5,985.19 \\
Manganese (mean, ppb)& 69.20 & 44.04 & 76.25 \\
Sodium (mean, ppb) & 21,654.68 & 11,169.23 & 30,845.61 \\
Zinc (mean, ppb) & 1,105.54 & 123.03 & 32,038.27 \\
\hline
Age (median, years) & 41.63 & 41.50 & 7.10 \\
Black Pop. Pct. & 0.20 & 0.14 & 0.19 \\
Female Pop. Pct. & 0.51 & 0.51 & 0.04 \\
Hisp Pop. Pct. & 0.08 & 0.06 & 0.08 \\
Income (median, $1,000)$& 58.6 & 52.5 & 25.9 \\
Population & 4,917.26 & 4,616.00 & 2,146.73 \\
Obesity Pct. & 34.19 & 34.30 & 5.18 \\
Well Water Users, Pct. & 43.53 & 38.80 & 35.43 \\
\bottomrule
\end{tabular}
\end{table}

Significant interactions according to the PAGP tests are displayed in Table \ref{tab:ncwell_interactions}.

\begin{table}[!htbp]
\centering
\caption{Interaction effect hypothesis test p-values for significant (p $<$ 0.05) interactions based on the PAGP interaction tests, for pairs of exposures and covariates identified as significant by the PAGP in the NCWELL dataset. Non-significant interactions are omitted because of the large number tested. Because our focus is on hypothesis testing in nonparametric models, we report p-values rather than effect estimates and confidence intervals. Demographic effects reflect the small area estimation methodology so are not otherwise discussed in our results.}
\label{tab:ncwell_interactions}
\begin{tabular}{lr}
\toprule
Interaction & p-value \\
\midrule
Age $\times$ Median Income& 0.017\\
Age $\times$ Hispanic Pct & $<$ 0.001 \\
Age $\times$ Well Use Pct & 0.041\\
Median Income $\times$ Black Pct & $<$ 0.001 \\
Median Income $\times$ Hispanic Pct & 0.014\\
Median Income $\times$ Arsenic & 0.035\\
Black Pct $\times$ Hispanic Pct& 0.046\\
Arsenic $\times$ Lead& 0.044\\
\bottomrule
\end{tabular}
\end{table}

The GAM hypothesis test output is in Table \ref{tab:ncwell_gam}.

\begin{table}[!htbp]
\centering
\caption{Main and interaction effect hypothesis test p-values from the secondary GAM used to interpret findings from the PAGP hypothesis test results, for the NCWELL data. The GAM included tensor interaction (ti) terms for main and interaction effects identified as significant by the PAGP. Because our focus is on hypothesis testing in nonparametric models, we report p-values rather than effect estimates and confidence intervals. Demographic effects reflect the small area estimation methodology so are not otherwise discussed in our results.}
\label{tab:ncwell_gam}
\begin{tabular}{lr}
\toprule
Term & p-value \\
\midrule
Latitude & $<$ 0.001 \\
Longitude & $<$ 0.001 \\
Age & $<$ 0.001 \\
Median Income & $<$ 0.001 \\
Black Pct & $<$ 0.001 \\
Hispanic Pct & $<$ 0.001 \\
Well Use Pct & $<$ 0.001 \\
Arsenic & 0.027\\
Lead &  0.023\\
Latitude $\times$ Longitude & $<$ 0.001\\
Arsenic $\times$ Lead & 0.241\\
Age $\times$ Hispanic Pct & $<$ 0.001 \\
Age $\times$ Well Use Pct & 0.651\\
Median Income $\times$ Black Pct & $<$ 0.001 \\
Median Income $\times$ Hispanic Pct & 0.352\\
Median Income $\times$ Arsenic & 0.088\\
Black Pct $\times$ Hispanic Pct& $<$0.001\\
\end{tabular}
\end{table}

The partial interaction effect between arsenic and median income is displayed in Figure \ref{fig:as_income_2d}.
\begin{figure}[H]

    \centering
    \includegraphics[width=0.95\linewidth]{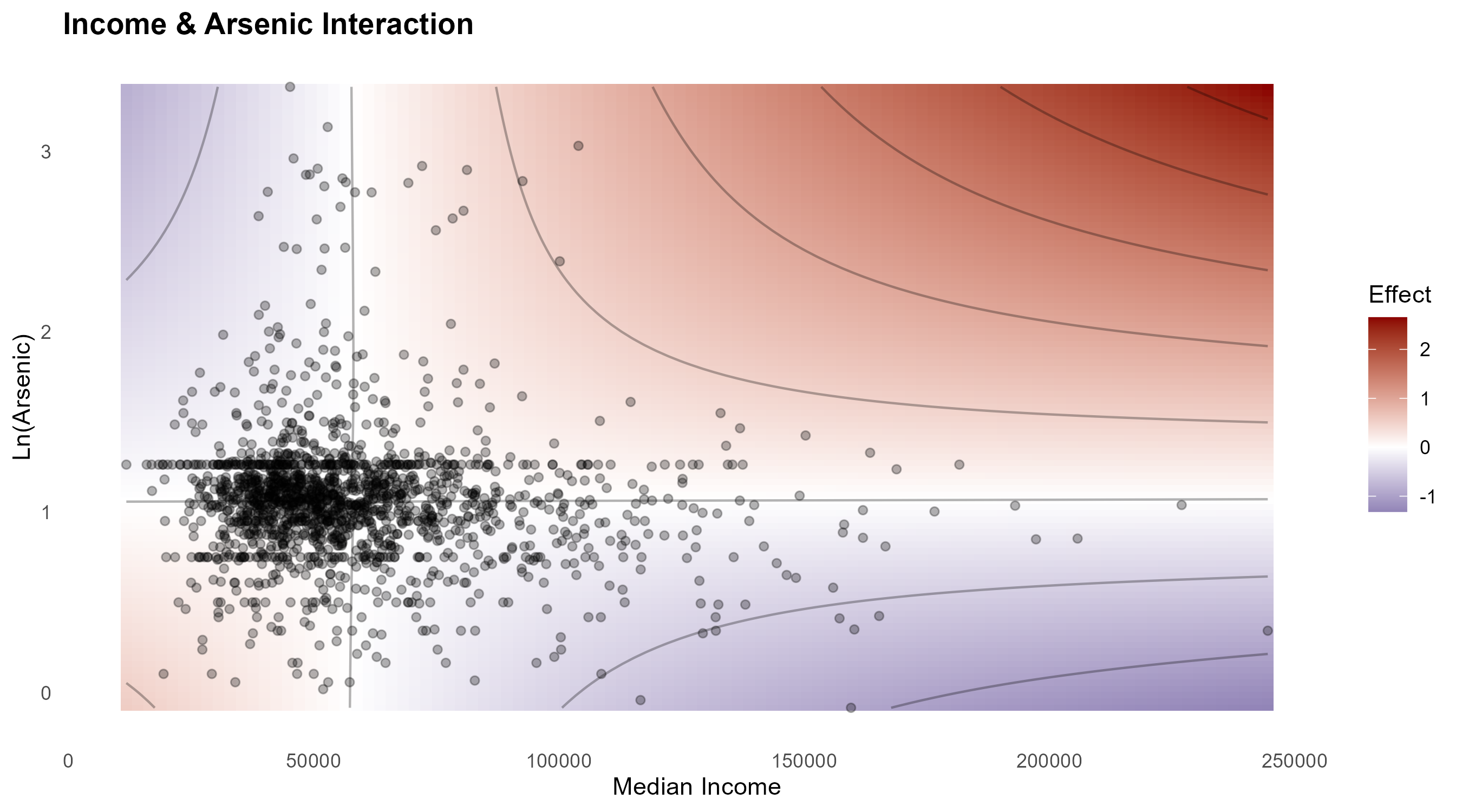}
    \caption{Partial interaction effect between arsenic and median income estimated by a GAM tensor interaction smooth term, using main and interaction effects identified by the PAGP hypothesis tests plus all covariates, for the NCWELL data. The partial effects characterize associations found between natural log-transformed metals in private wells, averaged by census tract, and small area estimates of obesity prevalence.}
        \label{fig:as_income_2d}

\end{figure}

Estimated main effects for several demographic covariates are displayed in Figure \ref{fig:ncwell_demo}, demonstrating highly nonlinear associations with small area estimates, necessitating a nonlinear model.

\begin{figure}[H]

    \centering
    \includegraphics[width=0.5\linewidth]{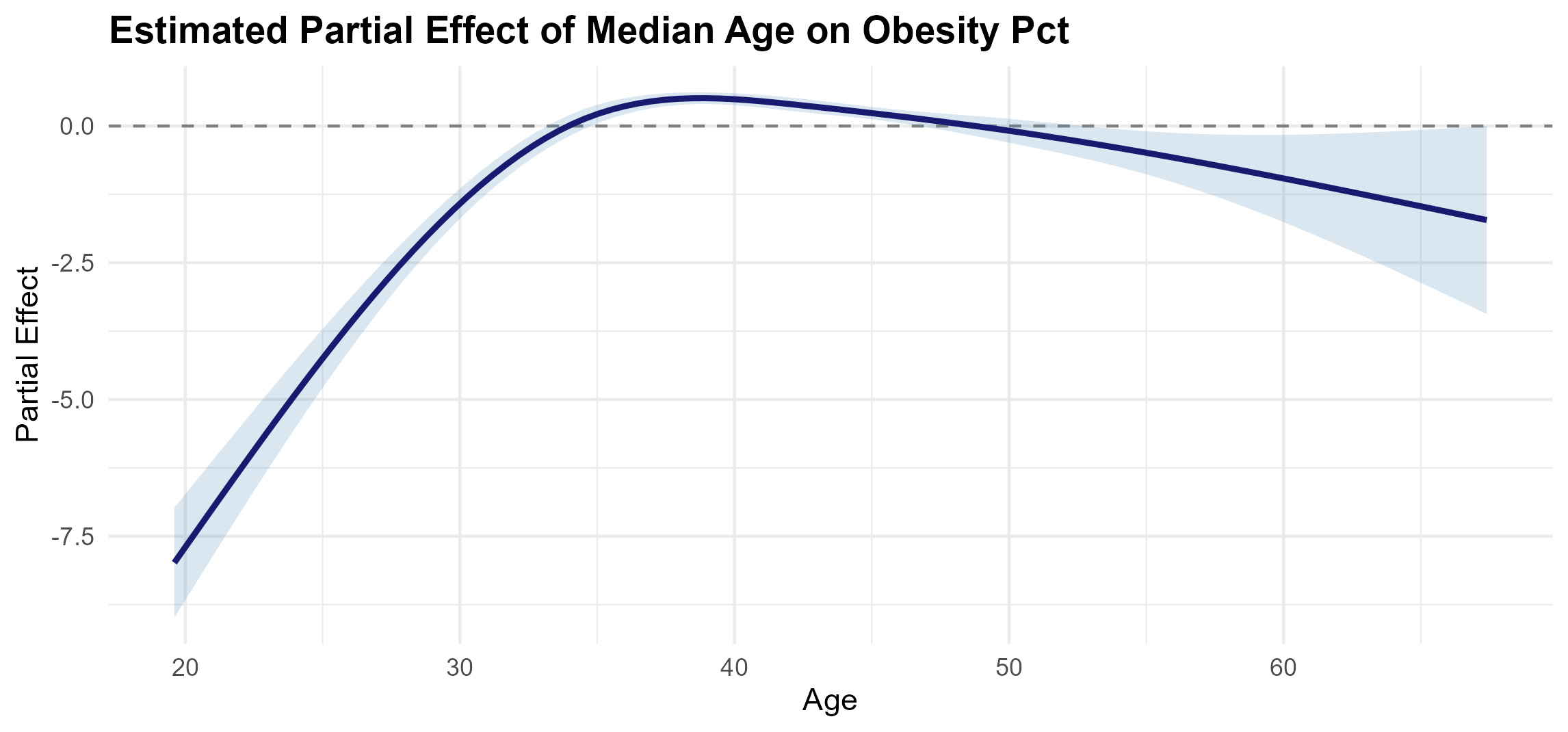}
    \includegraphics[width=0.5\linewidth]{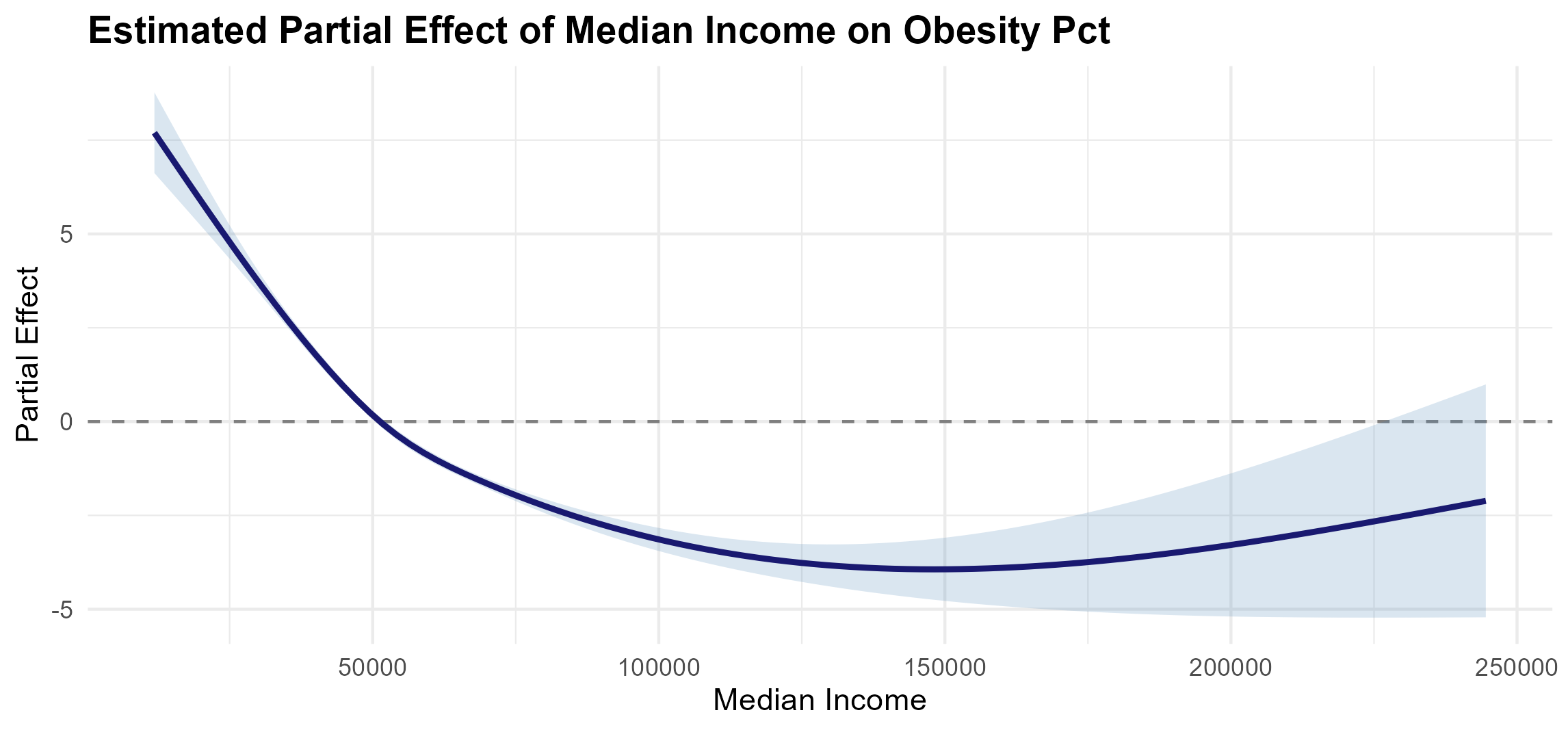}
    \caption{Partial main effects of median age and median income estimated by a GAM tensor interaction smooth term, using main and interaction effects identified by the PAGP hypothesis tests plus all covariates, for the NCWELL data. These partial effects characterize associations between demographic factors and small area estimates of obesity prevalence; their high nonlinearity and significance is due to the small area estimates' methodology relying heavily on demographic factors.}
        \label{fig:ncwell_demo}

\end{figure}

\subsection{NHANES 2017-2020 liver data}\label{s:nhanes_addl}

\begin{table}[!htbp]
\caption{NHANES 2017-2020 Liver Data, Categorical Variable Summary}
\centering
\label{tab:cat_summary}
\begin{tabular}{lrr}
\toprule
 & Count & Percent (\%) \\
\hspace{1em} Non-alcohol user & 1,159 & 46.2\% \\
\hspace{1em} Alcohol User & 1,348 & 53.8\% \\
\hline
\hspace{1em} No diabetes & 2,061 & 82.2\% \\
\hspace{1em} Diabetes & 446 & 17.8\% \\
\hline
\hspace{1em} No hypertension & 1,421 & 56.7\% \\
\hspace{1em} Hypertension & 1,086 & 43.3\% \\
\hline
\hspace{1em} Female & 1,032 & 41.2\% \\
\hspace{1em} Male & 1,475 & 58.8\% \\
\hline
\hspace{1em} Mexican American & 251 & 10.0\% \\
\hspace{1em} Non-Hispanic Asian & 126 & 5.0\% \\
\hspace{1em} Non-Hispanic Black & 595 & 23.7\% \\
\hspace{1em} Non-Hispanic White & 1,164 & 46.4\% \\
\hspace{1em} Other Hispanic & 198 & 7.9\% \\
\hspace{1em} Other Race - Including Multi-Racial & 173 & 6.9\% \\
\hline
\hspace{1em} Non-smoker & 1,459 & 58.2\% \\
\hspace{1em} Smoker & 1,048 & 41.8\% \\
\bottomrule
\end{tabular}
\end{table}

\begin{table}[!htbp]
\centering
\caption{NHANES 2017-2020 Liver Data, Continuous Variable Summary}
\label{tab:cont_summary}
\begin{tabular}{lrrr}
\toprule
 & Mean & Median & SD \\
\midrule
ALT (U/L) & 22.57 & 18.00 & 17.58 \\
Lead ($\mu$g/dL) & 1.35 & 1.05 & 1.19 \\
Manganese ($\mu$g/L) & 9.37 & 8.77 & 3.54 \\
Mercury ($\mu$g/L) & 1.18 & 0.62 & 1.86 \\
Cadmium ($\mu$g/L) & 0.70 & 0.45 & 0.74 \\
Selenium ($\mu$g/L) & 186.08 & 183.83 & 27.13 \\
Age (years) & 53.24 & 55.00 & 16.67 \\
BMI & 30.26 & 29.20 & 7.48 \\
Energy Intake (kcal) & 2,234.09 & 2,045.00 & 1,101.51 \\
HDL-C (mg/dL) & 52.60 & 50.00 & 16.92 \\
Poverty Ratio & 2.40 & 2.00 & 1.56 \\
Total Cholesterol (mg/dL) & 185.89 & 182.00 & 42.80 \\
\bottomrule
\end{tabular}
\end{table}


\begin{table}[!htbp]
\centering
\caption{Main and interaction effect hypothesis test p-values from the secondary GAM used to interpret findings from the PAGP hypothesis test results, for the NHANES 2017-2020 liver data. The GAM included tensor interaction (ti) terms for main and interaction effects identified as significant by the PAGP, plus main effects for remaining covariates. Because our focus is on hypothesis testing in nonparametric models, we report p-values rather than effect estimates and confidence intervals. Demographic effects reflect the small area estimation methodology so are not otherwise discussed in our results.}
\label{tab:gam_liver_data}
\begin{tabular}{lr}
\toprule
Variable & p-value \\
\midrule
Manganese & 0.267\\
Lead & 0.008\\
Cadmium & 0.112\\
Selenium & 0.002\\
Age & $<$ 0.001\\
HDL-C & $<$ 0.001\\
Total Cholesterol &0.007\\
Manganese $\times$ Male & 0.006\\
Manganese $\times$ HDL-C & 0.006\\
Lead $\times$ HDL-C & 0.145\\
Manganese $\times$ Hypertension & 0.184\\
Cadmium $\times$ HDL-C&  0.008\\
Selenium $\times$ HDL-C& 0.232\\
Age $\times$ Male & $<$ 0.001\\
Age $\times$ HDL-C & 0.013\\
Age $\times$ Total Cholesterol & $<$ 0.001\\
Age $\times$ Hypertension & 0.003\\
\bottomrule
\end{tabular}
\end{table}

\end{document}